%% file: main.tex
\documentclass[11pt]{article}
\input{comandi.tex}

\newtheorem{theorem}{Theorem}[section]
\newtheorem{lemma}[theorem]{Lemma}
\newtheorem{definition}[theorem]{Definition}
\newtheorem{proposition}[theorem]{Proposition}

\newtheorem{corollary}[theorem]{Corollary}

\newtheorem{example}[theorem]{Example}

\newtheorem{remark}[theorem]{Remark}

\usepackage[dvipsnames]{xcolor}

\title{Ideal Membership Problem for Boolean Minority and Dual Discriminator\footnote{This is the full version of the paper \cite{bharathi_et_al:MFCS2021} that appeared in the proceedings of the 46th International Symposium on Mathematical Foundations of Computer Science (MFCS 2021)
}}
\author{Arpitha P. Bharathi~\thanks{{arpitha.bharathi@chalmers.se}. Department of Mathematical Sciences, Chalmers University of Technology and the University of Gothenburg, SE-412 96 Gothenburg, Sweden.} \and Monaldo Mastrolilli~\thanks{{monaldo.mastrolilli@supsi.ch}. SUPSI - IDSIA (Istituto Dalle Molle di Studi sull'Intelligenza Artificiale), CH-6962 Lugano-Viganello, Switzerland.}}

\begin{document}
\date{}

\maketitle
\begin{abstract}
We consider the polynomial Ideal Membership Problem (IMP) for ideals encoding combinatorial problems that are instances of CSPs over a finite language. In this paper, the input polynomial $f$ has degree at most $d=O(1)$ (we call this problem IMP$_d$).  
We bridge the gap in \cite{MonaldoMastrolilli2019} by proving that the IMP$_d$ for Boolean combinatorial ideals whose constraints are closed under the minority polymorphism can be solved in polynomial time. 
This  completes the identification of the tractability for the Boolean $\IMP_d$. 
We also prove that the proof of membership for the $\IMP_d$ for problems constrained by the dual discriminator polymorphism over any finite domain can be found in polynomial time. 
Our results can be used in applications such as Nullstellensatz and Sum-of-Squares proofs.
\end{abstract}
\\
\subjectclass{13P10, 13F20}

\section{Introduction}
The study of polynomial ideals and related algorithmic problems goes back to David Hilbert~\cite{Hilbert1893}. The methods developed in this area to date find a wide range of applications in mathematics and computer science.
In this paper we consider the polynomial Ideal Membership Problem, where we want to decide if a given polynomial belongs to a given ideal. 
This is a fundamental algorithmic problem with important applications in solving polynomial systems (see e.g.~\cite{Cox}), polynomial identity testing \cite{Cox,PIT} and underlies proof systems such as Nullstellensatz and Polynomial Calculus (see e.g. \cite{BeameIKPP94,BussP98,Grigoriev98}).

To introduce the problem formally, let $\Field[x_1, \ldots, x_n]$ be the ring of polynomials over a field $\Field$ with indeterminates $x_1,\ldots, x_n$. 
A polynomial \emph{ideal} $I$ is a subset of the polynomial ring $ \mathbb{F}[x_1,\dots,x_n]$ with two properties: for any two polynomials $f,g$ in $I$, $f+g$ also belongs to $I$ and so does $hf$ for any polynomial $h\in \mathbb{F}[x_1,\dots,x_n]$. 
By the Hilbert Basis Theorem \cite{HilbertBasisTheorem} every ideal $I$ has a finite generating set $F=\{f_1,\ldots,f_r\}\subset I$ such that for every $f\in \Field[x_1, \ldots, x_n]$, we have $f\in I$ if and only if there is an ``ideal membership proof'', namely a set of polynomials $\{h_1,\ldots,h_r\} \subset \Field[x_1, \ldots, x_n]$ such that $f=h_1f_1+\ldots+h_rf_r$. 
%
The polynomial \textsc{Ideal Membership Problem} ($\IMP$) is to find out if a polynomial $f$ belongs to an ideal $I$ or not, given a set $F$ of generators of the ideal (we use $\IMP_d$ to denote $\IMP$ when the input polynomial $f$ has degree at most $d=O(1)$). 
The $\IMP$ is, in general, notoriously intractable. The results of Mayr and Meyer show that it is EXPSPACE-complete~\cite{Mayr1989,MAYR1982305}.

Semidefinite programming (SDP) relaxations have been a powerful technique for approximation algorithm design ever since the celebrated Goemans and Williamson result of Max-Cut \cite{goemans1995improved}. With the aim to construct stronger and stronger SDP relaxations, the Sum-of-Squares ($\sos$) hierarchy has emerged as the most promising set of relaxations (see, e.g., \cite{Laurent2009}). However, we still do not know the answer to even very basic questions about its power. For example, we do not even know when $\sos$ is guaranteed to run in polynomial time. As recently observed by O’Donnell~\cite{odonnell2017}, a bounded degree $\sos$ proof does not necessarily imply its low bit complexity, showing that the often repeated claim that for any fixed degree $\sos$ runs in polynomial time, is far from true. O’Donnell raised the open problem to establish useful conditions under which “small” $\sos$ proofs can be guaranteed. With this aim, a first elegant sufficient condition is due to Raghavendra and Weitz \cite{raghavendra_weitz2017,Weitz:Phd}. For each instance $\mathcal{C}$ of a combinatorial problem, the set of polynomials that vanish at every point of the set of solutions of $\mathcal{C}$ is called the \emph{combinatorial ideal} of $\mathcal{C}$ and denoted by $I_\mathcal{C}$.
To satisfy Raghavendra and Weitz's criterion, it is necessary (but also sufficient) that the ideal membership problem $\IMP_d$ for each $I_\mathcal{C}$ is polynomial-time solvable and
that ideal membership proofs can be efficiently found too.\footnote{Note that answering whether a polynomial belongs to a certain ideal does not necessarily mean finding an ideal membership proof of that.}
So the tractability of the ideal membership proof ensures that $\sos$ runs in polynomial time for combinatorial problems. This is currently the only known general result that addresses the $\sos$ bit complexity issue. However, the $\IMP_d$ tractability criterion of Raghavendra and Weitz  suffers from a severe limitation, namely it is not clear which restrictions on combinatorial problems can guarantee an efficient computation of the $\IMP_d$ proofs for combinatorial ideals.

The Constraint Satisfaction Problem (CSP) provides a general framework for a
wide range of combinatorial problems, where we are given a set of variables and a set of constraints, and we have to decide whether the variables can be assigned values that satisfy the constraints.
There are useful connections between the $\IMP_d$ and the $\CSP$: for example a $\CSP$ instance $\Cc$ is unsatisfiable if and only if $1\in\I_\Cc$. 
It follows that the $\CSP$ is just the special case of the $\IMP_d$ with $d=0$ (see \cref{sect:idealCSP} for more details on Ideal-CSP correspondence). 
Restrictions on CSPs, called $\CSP(\Gamma)$, in which the type of constraints is limited to relations from a set $\Gamma$, have been successfully applied to study the computational complexity classification (and other algorithmic properties) of $\CSP$s (see \cite{2017dfu7} for an excellent survey). 

Motivated by the aforementioned issue of Raghavendra and Weitz criterion, Mastrolilli~\cite{MonaldoMastrolilli2019} initiated a systematic study of the $\IMP_d$ tractability for combinatorial ideals of the form $\IMP_d(\Gamma)$ arising from combinatorial problems from $\CSP(\Gamma)$ for a set of relations $\Gamma$ over the Boolean domain. 
The classic dichotomy result of Schaefer \cite{Schaefer78} gives the complexity of $\CSP(\Gamma)$ (and therefore of $\IMP_0(\Gamma)$) for the Boolean domain:
$\CSP(\Gamma)$ is solvable in polynomial time if all constraints are closed under one of six polymorphisms (majority, minority, MIN, MAX, constant 0 and constant 1), else it is NP-complete.
Mastrolilli \cite{MonaldoMastrolilli2019} claimed a dichotomy result for the $\IMP_d(\Gamma)$ for the Boolean domain: for any constant $d\geq 1$, the $\IMP_d(\Gamma)$ of Boolean combinatorial ideals is solvable in polynomial time if all constraints are closed under one of four polymorphisms (majority, minority, MIN, MAX), else it is coNP-complete.
In \cite{MonaldoMastrolilli2019}, it is shown that $\IMP_d(\Gamma)$ is polynomial time solvable for three polymorphisms (majority, MIN, MAX), and moreover ideal membership proofs can be efficiently found too.
Whereas for the ternary minority polymorphism it was incorrectly declared to have been resolved by a previous result\footnote{This was pointed out by Andrei Bulatov, Akbar Rafiey and Stanislav \v{Z}ivn\'{y}.}. 
It was mistakenly assumed in \cite{MonaldoMastrolilli2019} that computing the (mod 2) \GB basis in lexicographic order was sufficient to solve the $\IMP_d$ problem in polynomial time, but the issue is that we require polynomials to be over $\mathbb{Q}$ and not $\textrm{GF}(2)$. 

We address these issues in this paper and, therefore establish the full dichotomy result claimed in \cite{MonaldoMastrolilli2019} (see \cite{mastrolilli_talg21} for an updated version of the paper). To ensure efficiency of the $\IMP_d(\Gamma)$ for the ternary minority polymorphism, it is sufficient to compute a $d$-truncated \GB basis in the graded lexicographic order (see \cref{def:dTruncated GB,sect:background} for more details).
This is achieved by first showing that we can easily find a \GB basis in the lexicographic order for the combinatorial ideal. Since polynomials in this \GB basis can have degrees up to $n$, we show how this basis can be converted to a $d$-truncated \GB basis in the graded lexicographic order in polynomial time, whose polynomials have degrees up to $d$ and coefficients of constant size. This efficiently solves $\IMP_d(\Gamma)$ for combinatorial ideals whose constraints are over a language $\Gamma$ closed under the minority polymorphism. Together with the results in~\cite{MonaldoMastrolilli2019,mastrolilli_talg21}, our result allows us to identify the tractability of the Boolean $\IMP_d(\Gamma)$.
Thus the following summarizes our first result of this paper:

\begin{theorem}\label{thm:MainTheorem}
Let $\Gamma$ be a constraint language over the Boolean domain that is closed under the minority polymorphism. For each instance $\Cc$ of $\CSP(\Gamma)$, 
the $d$-truncated reduced \GB basis in the graded lexicographic monomial ordering of the combinatorial ideal $\I_\Cc$ can be computed in $n^{O(d)}$ time.
\end{theorem}
\begin{corollary}\label{thm:Corollary_minority}
  If $\Gamma$ is closed under the minority polymorphism, then the ideal membership proofs for $\IMP_d(\Gamma)$ over the Boolean domain can be computed in polynomial time for $d=O(1)$.
\end{corollary}
After the appearance of a preliminary version of this paper~\cite{bharathi2020ideal}, Bulatov and Rafiey~\cite{bulatov2020complexity} obtained exciting results.
For a finite domain $D=\{0,1,\dots,p-1\}$ with prime $p$ elements, they consider the affine polymorphism $\otimes:D^3\rightarrow D$ defined as $\otimes (a,b,c)= a - b + c \Mod{p}$ (this generalizes the Boolean minority operation). By building on our approach,  they prove that a $d$-truncated \GB basis can be computed in time $n^{O(d)}$ for any fixed prime $p$. They then generalized the result to IMPs over languages invariant under affine operations of arbitrary finite Abelian groups ~\cite{bulatov_abeliangroups}.

In \cite{bharathi_et_al:MFCS2020, Bharathi_etal_SIAM2022}, we began the generalization of $\CSP(\Gamma)$ (viz. $\IMP_0(\Gamma)$) by working on the corresponding $\IMP_d(\Gamma)$ for any $d=O(1)$ in the three-element domain, which expands the known set of tractable $\IMP_d$ cases by providing a suitable class of combinatorial problems. We considered problems constrained under the dual discriminator polymorphism and proved that we can find the reduced \GB basis of the corresponding combinatorial ideal in polynomial time. This ensures that we can check if any degree $d$ polynomial belongs to the combinatorial ideal or not in polynomial time, and provide proof of membership if it does.
Among the results obtained in \cite{bulatov2020complexity}, the authors show that the $\IMP_d$ is solvable in polynomial time for \emph{any} finite domain for problems constrained under the dual discriminator. 
This was done by eliminating permutation constraints in some sense through a preprocessing step and converting an instance $\Cc = (X,D,C)$ for $\CSP(\Gamma)$ (where $X$ is a set of variables and $C$  is a set of constraints over $\Gamma$ with variables from $X$, see \cref{def:csp} for notation) to an instance $\Cc'=(X',D,C')$ where $X'\subseteq X$ and $C'\subseteq C$. 
Moreover, a polynomial $f(X)$ was converted to a polynomial $f'(X')$ such that $f\in\I_\Cc$ if and only if $f'\in\I_{\Cc'}$. They calculated a \GB basis of $I_{\Cc'}$ in polynomial time, which considered the remaining constraints. This gives a proof of membership of $f'$ in $I_{\Cc'}$ if it does belong to the ideal, but it was not yet known as to how to recover the proof of membership for $f$ in $\I_\Cc$. Meanwhile, our results in \cite{bharathi_et_al:MFCS2020, Bharathi_etal_SIAM2022} give a proof of membership, but are only constrained to a 3-element domain. 

In this paper, we compute a \GB basis for the entire combinatorial ideal over a finite domain by showing that a \GB basis of the ideal associated with permutation constraints can also be calculated in polynomial time. We forego the pre-processing step of \cite{bulatov2020complexity}, include the permutation constraints and directly calculate a \GB basis of $\I_\Cc$. The set of polynomials that the elements of the \GB basis can come from is polynomial in size, and hence we show that a proof of membership can also be calculated in polynomial time as required in \cite{RaghavendraW17}.
The following summarizes the second result of this paper:
\begin{theorem}\label{thm:main_dualdisc}
Let $\Gamma$ be a constraint language over a finite domain $D$ that is closed under the dual discriminator polymorphism.
For each instance $\Cc$ of $\CSP(\Gamma)$, a \GB basis in the graded lexicographic monomial ordering of the combinatorial ideal $\I_\Cc$ can be computed in time polynomial in the number of variables. Polynomials in this basis have degrees at most $|D|!$.
\end{theorem}
\begin{corollary}\label{thm:corollary_dualdisc}
  If $\Gamma$ is closed under the dual discriminator polymorphism, then membership proofs for $\IMP_d(\Gamma)$ over a finite domain can be computed in polynomial time.
\end{corollary}

After the conference version of this paper \cite{bharathi_et_al:MFCS2021}, Bulatov and Rafiey have updated their results \cite{bulatov2022STOCupdated} to show that a $d$-truncated \GB basis can be obtained for several polymorphisms, including the dual discriminator. Although their approach is more versatile it does not show that a full \GB basis can be obtained in $n^{O(|D|!)}$ time and has polynomial size independent of $d$.

The study of CSP-related $\IMP$s is in its early stages. The results obtained in this paper are steps towards the long-term and challenging goal of extending the celebrated dichotomy results of $\CSP(\Gamma)$ for finite domain \cite{Bulatov17,Zhuk17} to $\IMP_d(\Gamma)$. This would provide a complete CSP-related characterization of when the IMP tractability criterion is applicable.

To make the paper self-contained, some essential background and standard (according to the book \cite{Cox}) \GB basis notations can be found in the following section.

\section{Preliminaries}\label{sec:preliminaries}
\subsection{Ideals and varieties}\label{sect:background}
Let $\Field$ denote an arbitrary field (for the applications of this paper $\Field=\mathbb{Q}$). Let $\Field[X]$ be the ring of polynomials over a field $\Field$ and the set of indeterminates $X=\{x_1,\ldots, x_n\}$. Let $\Field[X]_d$ denote the subspace of polynomials of degree at most $d$.

\begin{definition}\label{def:ideal}
The ideal (of $\Field[X]$) generated by a finite set of polynomials $\{f_1,\dots, f_m\}$ in $\Field[X]$ is defined as
\[\Ideal{ f_1,\ldots,f_m}:= \left\{\sum_{i=1}^m h_i f_i\mid h_1,\ldots,h_m\in \Field[X]\right\}.\] 

The set of polynomials that vanish on a given set $S\subset \Field^n$ is called the \textbf{\emph{vanishing ideal}} of $S$ and is defined as
$\Ideal{S}:= \{f\in \Field[X]\mid f(a_1,\ldots,a_n)=0, \forall (a_1,\ldots,a_n)\in S\}$.
\end{definition}

\begin{definition}
An ideal $I$ is \textbf{\emph{radical}} if $f^m \in I$ for some integer $m\geq 1$ implies that $f\in  I$.
\end{definition}

Other ways to denote $\Ideal{ f_1,\ldots,f_m}$ is by $\Ideal{ \{f_1,\ldots,f_m\}}$, $\langle f_1,\ldots,f_m \rangle$ and $\langle \{ f_1,\ldots,f_m\} \rangle$, and we will use notations interchangeably.
\begin{definition}\label{def:V(I)}
Let $\{f_1,\ldots, f_m\}$ be a finite set of polynomials in $\Field[X]$. The \textbf{\emph{affine variety}} defined by $f_1,\ldots, f_m$ is
$\Variety{ f_1,\ldots,f_m}:= \{(a_1,\ldots,a_n)\in \Field^n|\; f_i(a_1,\ldots,a_n)=0, \quad 1\leq i\leq m\}$.
Let $I\subseteq \Field[X]$ be an ideal. We will denote by $\mathbf{V}(I)$ the set $\Variety{I}=\{(a_1,\ldots,a_n)\in \Field^n\mid f(a_1,\ldots,a_n)=~0, \quad \forall f\in I\}$.
\end{definition}

\begin{theorem}[\cite{Cox}, Th.15, p.196]\label{th:ideal_intersection}
  If $I$ and $J$ are ideals in $\Field[X]$, then $\Variety{I\cap J}= \Variety{I}\cup \Variety{J}$.
\end{theorem}

\begin{definition}\label{def:ideal_sum}
  If $I$ and $J$ are ideals in $\Field[X]$, then the sum of $I$ and $J$, denoted by $I+J$, is the ideal $I+J=\{f+g\mid f\in I \textrm{ and } g\in J\}$.
\end{definition}

\begin{corollary}[\cite{Cox}, Cor.3, p.190] \label{cor:ideal_sum}
    If $f_1,\dots,f_r \in \Field[X]$, then $\langle f_1,\dots, f_r \rangle = \langle f_1 \rangle + \dots + \langle f_r \rangle$.
 
\end{corollary}

\begin{theorem}[\cite{Cox}, Th.4, p.190]\label{th:ideal_sum}
  If $I$ and $J$ are ideals in $\Field[X]$, then $\Variety{I+J}= \Variety{I}\cap \Variety{J}$.
\end{theorem}

\subsection{CSP and polymorphisms}\label{sec:CSPandpolymorphisms} 
Let $D$ denote a finite set called the  \textbf{\emph{domain}}.
By a $k$-ary \textbf{\emph{relation}} $R$ on a domain $D$ we mean a subset of the $k$-th Cartesian power $D^k$; $k$ is said to be the \textbf{\emph{arity}} of the relation. We often use relations and (affine) varieties interchangeably since both are subsets of $D^k$ (we will not refer to varieties from universal algebra in this paper). A \textbf{\emph{constraint language}} $\Gamma$ over $D$ is a finite set of relations over $D$. A constraint language is 
\textbf{\emph{Boolean}} if it is over the two-element domain $\{0,1\}$. 

A \emph{\textbf{constraint}} over a constraint language $\Gamma$ is an expression of the form $R(x_{i_1},\ldots, x_{i_k})$ where $R$ is a relation of arity $k$ contained in $\Gamma$, and $x_{i_1},\ldots, x_{i_k}$ are variables that belong to the variable set $X$. A constraint is satisfied by a mapping $\phi:X\rightarrow D$ if $(\phi(x_{i_1}),\ldots, \phi(x_{i_k}))\in R$.

\begin{definition}\label{def:csp}
 The \emph{(nonuniform) \textsc{Constraint Satisfaction Problem} ($\CSP$)} associated with language $\Gamma$ over $D$ is the problem $\CSP(\Gamma)$ in which: an instance is a triple $\Cc=(X,D,C)$ where $X=\{x_1,\ldots,x_n\}$ is a set of $n$ variables and $C$  is a set of constraints over $\Gamma$ with variables from $X$. The goal is to decide whether or not there exists a solution, i.e.~a mapping $\phi: X\rightarrow D$ satisfying all of the constraints. We will use $Sol(\Cc)\subseteq D^n$ to denote the set of solutions of $\Cc$.
\end{definition}
Moreover, we follow the algebraic approach to Schaefer's dichotomy result \cite{Schaefer78} formulated by Jeavons \cite{JEAVONS1998185} where each class of CSPs that are polynomial time solvable is associated with a polymorphism.
\begin{definition}\label{def:polymorph}
An operation $f:D^m \rightarrow D$ is a \textbf{polymorphism} of a relation $R\subseteq D^k$ if for any choice of $m$ tuples $(t_{11},\dots,t_{1k}),\dots,(t_{m1},\dots,t_{mk})$ from $R$ (allowing repetitions), it holds that the tuple obtained from these $m$ tuples by applying $f$ coordinate-wise, $(f(t_{11},\dots,t_{m1}),\dots,f(t_{1k},\dots,t_{mk}))$, is in $R$. We also say that $f$ \textbf{preserves} $R$, or that $R$ is \textbf{invariant} or \textbf{closed} with respect to $f$. A polymorphism of a constraint language $\Gamma$ is an operation that is a polymorphism of every $R\in \Gamma$.
\end{definition}
In this paper, we deal with two particular polymorphisms: the Boolean minority (\cref{def:minority}) and the dual discriminator (\cref{def:dual discriminator}).
\begin{example}\label{example:CSP}
Consider $D=\{0,1,2\}$ and relations $R_1=\{(0,1,1),(2,0,2),(2,2,1), (2,0,1)$, $(2,1,1)\}$ and $R_2=\{(1,1),(2,1)\}$. Let the language $\Gamma = \{R_1,R_2\}$. Consider the instance $\mathcal{C}=(X,D,C)$ of $\CSP(\Gamma)$ with $X=\{x,y,z\}$ and $C=\{C_1 = R_1 (x,y,z),C_2= R_2(x,z)\}$. The assignment $\phi$ where $\phi(x)=2, \phi(y)= 0, \phi(z)=1$ is a solution to the instance $\mathcal{C}$. 
\end{example}

\subsection{The ideal-CSP correspondence}\label{sect:idealCSP}
Let $\Cc=(X,D,C)$ be an instance of $\CSP(\Gamma)$ (see ~\cref{def:csp}). Without loss of generality, we shall assume that $D\subseteq \Field$.
Let $Sol(\Cc)$ be the (possibly empty) set of all feasible solutions of $\Cc$.
A simple construction shows that we can map $Sol(\Cc)$ to a zero-dimensional and radical ideal $\I_\Cc\subseteq \Field[X]$ such that $Sol(\Cc)=\Variety{I_\Cc}$. We refer to \cite{mastrolilli_talg21} for details.

In particular, let $Y=(x_{i_1},\ldots,x_{i_k})$ be a $k$-tuple of variables from $X$
and let $R(Y)$ be a non empty constraint from $C$.
$R(Y)$ can be mapped to a generating system of an ideal such that the projection of the variety of this ideal onto $Y$ is equal to $R(Y)$ (see~\cite{vandongenPhd} for more details).
%
%
%
%
%
%
%
Furthermore, for each variable $x_k$, we will assume that the polynomial $\prod_{j\in D}(x_{i_k}-j)$ is in $\I_\Cc$. The latter forces $x_k$ to take values only in $D$. We call it the \textbf{\emph{domain polynomial}} denoted by $dom(x_k)$.

We emphasize that if there is no solution, then we have from Hilbert's Weak Nullstellensatz (see e.g.~\cite{Cox,mastrolilli_talg21}) that $1\in\I_\Cc \iff Sol(\Cc)=\emptyset \iff\I_\Cc = \mathbb{Q}[X]$. 

\subsection{\GB bases}\label{sect:GBbasics}
A monomial ordering $>$ on $\Field[X]$ is a total ordering on the set of all monomials in $\Field[X]$, that is invariant under multiplication (i.e.~if $a > b$ and $c$ is any other monomial then $ac>bc$, see~\cite{Cox}, Definition 1, p.55).
We identify the monomial $x^\alpha=x_1^{\alpha_1}\cdots x_n^{\alpha_n}$ with the $n$-tuple of exponents $\alpha =(\alpha_1,\ldots,\alpha_n)\in \Zz^n_{\geq0}$. This establishes a one-to-one correspondence between the monomials in $\Field[X]$ and $\Zz^n_{\geq0}$. Any ordering $>$ we establish on the space $\Zz^n_{\geq0}$
will give us an ordering on monomials: if $\alpha > \beta$ according to this ordering, we also say that $x^\alpha > x^\beta$. The monomial ordering that we use in this paper is the graded lexicographic ordering.

\begin{definition}\label{def:lex and grlex} Let $\alpha =(\alpha_1,\ldots,\alpha_n),\beta=(\beta_1,\ldots,\beta_n)\in \Zz^n_{\geq0}$ and $|\alpha| = \sum_{i=1}^n\alpha_i$, $|\beta| = ~\sum_{i=1}^n\beta_i$.
  \begin{enumerate}[(i)]
      \item Lexicographic order: We say $\alpha>_\lex \beta$ if, in the vector difference $\alpha -\beta \in \Zz^n$, the left most nonzero entry is positive. 
      \item Graded lexicographic order: We say $\alpha>_\grlex \beta$ if $|\alpha| >|\beta|$, or $|\alpha| =|\beta|$ and $\alpha>_\lex \beta$.

  \end{enumerate}
\end{definition}

\begin{definition}
  For any $\alpha=(\alpha_1,\ldots,\alpha_n)\in \Zz^n_{\geq0}$, let $x^\alpha:= \prod_{i=1}^{n}x_i^{\alpha_i}$. Let $f= \sum_{\alpha} a_{\alpha}x^\alpha$ be a nonzero polynomial in $\Field[X]$ and let $>$ be a monomial order (in this paper, this is always the \grlex order).
  \begin{enumerate}[(i)]
    
    \item The \textbf{\emph{multidegree}} of $f$ is $\multideg(f):= \max(\alpha\in \Zz^n_{\geq0}:a_\alpha\not = 0)$.
    \item The \textbf{\emph{degree}} of $f$ is deg$(f):=\sum_{i=1}^n \alpha_i$, where $(\alpha_1,\ldots,\alpha_n)=\multideg(f)$. 
    \item The \textbf{\emph{leading coefficient}} of $f$ is $\LC(f):= a_{\multideg(f)}\in \Field$.
    \item The \textbf{\emph{leading monomial}} of $f$ is $\LM(f):= x^{\multideg(f)}$ (with coefficient 1).
    \item The \textbf{\emph{leading term}} of $f$ is $\LT(f):= \LC(f)\cdot \LM(f)$.
  \end{enumerate}
\end{definition}
%

The concept of \emph{reduction}, also called \emph{multivariate division} or \emph{normal form computation}, is central to \GB basis theory. It is a multivariate generalization of the Euclidean division of univariate polynomials.

\begin{definition}\label{def:reduction}
Fix a monomial order and let $G=\{g_1,\ldots,g_t\}\subset \Field[X]$. Given $f\in \Field[X]$, we say that \emph{\textbf{$f$ reduces to $r$ modulo $G$}}, written
$f\rightarrow_G r$,
if $f$ can be written in the form
$f=A_1g_1+\dots+A_t g_t+r$ for some $A_1,\ldots,A_t,r\in \Field[X]$,
such that:
\begin{enumerate}[(i)]
  \item no term of $r$ is divisible by any of $\LT(g_1),\ldots,\LT(g_t)$ and 
  \item whenever $A_i g_i\not=0$, we have $\multideg(f)\geq \multideg(A_ig_i)$.
\end{enumerate}
The polynomial remainder $r$ is called a \emph{\textbf{normal form of $f$ by $G$}} and will be denoted by $f|_G$.
\end{definition}

A normal form of $f$ by $G$, i.e.~$f|_G$, can be obtained by repeatedly performing the following until it cannot be further applied: choose any $g\in G$ such that $\LT(g)$ divides some term $t$ of $f$ and replace $f$ with $f-\frac{t}{\LT(g)}g$. Note that the order we choose the polynomials $g$ in the division process is not specified.
In general a normal form $f|_G$ is not uniquely defined.
Even when $f$ belongs to the ideal generated by $G$, i.e.~$f\in \Ideal{G}$, it is not always true that $f|_G=0$. In particular, the order in which polynomial division is performed is important, as the following example shows.
\begin{example}
  Let $f=xy^2-y^3$ and $G=\{g_1,g_2\}$, where $g_1=xy-1$ and $g_2=y^2-1$. Consider the graded lexicographic order (with $x>y$). If we divide first by $g_2$ the quotient is $(x-y)$ and the remainder is $(x-y)$; the latter is not divisible by $g_1$: $f = 0\cdot g_1 + (x-y)\cdot g_2 + x-y$. Vice versa, if we divide first by $g_1$ and then by $g_2$ we obtain a zero remainder: $f = y\cdot g_1 - y\cdot g_2 + 0$.
\end{example}
This non-uniqueness is the starting point of \GB basis theory.
\begin{definition}
Fix a monomial order on the polynomial ring $\Field[X]$. A finite subset $G = \{g_1,\ldots, g_t\}$ of an ideal $I \subseteq \Field[X]$ different from $\{0\}$ is said to be a \emph{\textbf{\GB basis}} (or \emph{\textbf{standard basis}}) if
$\langle \LT(g_1),\ldots, \LT(g_t)\rangle = \langle \LT(I)\rangle$, where we denote by $\langle \LT(I)\rangle$ the ideal generated by the elements of the set $\LT(I)$ of leading terms of nonzero elements of $I$.
\end{definition}
\begin{definition}\label{def:redGB}
The \textbf{\emph{reduced \GB basis}} of a polynomial ideal $I$ is a \GB basis $G$ of $I$ such that:
\begin{enumerate}[(i)]
  \item $\LC(g)= 1$ for all $g \in G$ and
  \item for all $g \in G$, no monomial of $g$ lies in $\GIdeal{\LT(G\setminus \{g\})}$.
\end{enumerate}
\end{definition}
It is known (see~\cite{Cox}, Theorem~5, p.93) that for a given monomial ordering, a polynomial ideal $I\not=\{0\}$ has a unique reduced \GB basis. 
\begin{proposition}[\cite{Cox}, Proposition~1, p.83]\label{th:gbprop}
Let $I\subset \Field[X]$ be an ideal and let $G=\{g_1,\ldots,g_t\}$ be a \GB basis for $I$. Then given $f\in \Field[X]$, $f$ can be written in the form
$f=A_1g_1+\dots+A_t g_t+r$ for some $A_1,\ldots,A_t,r\in \Field[X]$,
such that:
\begin{enumerate}[(i)]
  \item no term of $r$ is divisible by any of $\LT(g_1),\ldots,\LT(g_t)$,
  \item whenever $A_i g_i\not=0$, we have $\multideg(f)\geq \multideg(A_ig_i)$ and
  \item $r$ is unique.
\end{enumerate}
In particular, $r$ is the remainder on division of $f$ by $G$ no matter how the elements of $G$ are listed when using the division algorithm.
\end{proposition}
\begin{corollary}[\cite{Cox}, Corollary~2, p.84]\label{th:imp}
  Let $G=\{g_1,\ldots,g_t\}$ be a \GB basis for $I\subseteq \Field[X]$ and let $f\in \Field[X]$. Then $f\in I$ if and only if the remainder on division of $f$ by $G$ is zero.
\end{corollary}
\begin{definition}\label{def:pdiv}
We will write $\reduce f F$ for the remainder of $f$ by the ordered $s$-tuple $F=(f_1,\ldots,f_s)$. If $F$ is a $\GB$ basis for $\spn{f_1,\dots,f_s}$, then we can regard $F$ as a set (without any particular order) by Proposition~\ref{th:gbprop}.
\end{definition}
%

%
The ``obstruction'' to $\{g_1,\ldots, g_t\}$ being a \GB basis is the possible occurrence of polynomial combinations of $g_i$ whose leading terms are not in the ideal generated by the $\LT( g_i)$. One way (actually the only way) this can occur is if the leading terms in a suitable combination cancel, leaving only smaller terms. The latter is fully captured by the so called $S$-polynomials that play a fundamental role in \GB basis theory.
\begin{definition}\label{def:spoly}
Let $f,g\in \Field[X]$ be nonzero polynomials. Let $\multideg(f)$  $=\alpha$ and $\multideg(g)=\beta$, with $\gamma=(\gamma_1,\ldots,\gamma_n)$, where $\gamma_i = \max(\alpha_i,\beta_i)$ for each $i$. We call $x^\gamma$ the \emph{\textbf{least common multiple}} of $\LM(f)$ and $\LM(g)$, written $x^\gamma = \LCM(\LM(f),\LM(g))$.
The \emph{\textbf{$S$-polynomial}} of $f$ and $g$ is the combination $S(f,g) = \frac{x^\gamma}{\LT(f)}\cdot f - \frac{x^\gamma}{\LT(g)}\cdot g$.
\end{definition}
The use of $S$-polynomials to eliminate leading terms of multivariate polynomials generalizes the row reduction algorithm for systems of linear equations. If we take a system of homogeneous linear equations (i.e.: the constant coefficient equals zero), then it is not hard to see that bringing the system in triangular form yields a \GB basis for the system.
\begin{theorem}[\textbf{Buchberger's Criterion}] (See e.g. \cite{Cox}, Theorem 3, p.105)\label{th:crit}
A basis $G=\{g_1,\ldots,g_t\}$ for an ideal $I$ is a \GB basis if and only if $S(g_i,g_j)|_G= 0$ for all $i\not=j$.
\end{theorem}
%
By Theorem~\ref{th:crit} it is easy to show whether a given basis is a \GB basis. Indeed, if $G$ is a \GB basis then given $f\in \Field[X]$, $f|_G$ is unique and it is the remainder on division of $f$ by $G$, no matter how the elements of $G$ are listed when using the division algorithm.
Furthermore, Theorem~\ref{th:crit} leads naturally to an algorithm for computing \GB bases for a given ideal $I=\langle f_1,\ldots,f_s \rangle$: start with a basis $G=\{f_1,\ldots,f_s\}$ and for any pair $f,g\in G$ with $S(f,g)|_G\not= 0$ add $S(f,g)|_G$ to $G$.
This is known as Buchberger's algorithm~\cite{BUCHBERGER2006475} (for more details see \cite[p. 90]{Cox} or \cite[Algorithm 1]{mastrolilli_talg21}).
In the remainder of the paper we will make repeated use of the following simple fact~\cite[Proposition 4 p. 106]{Cox}.
\begin{proposition}\label{th:rel_prime}
    We say the leading monomials of two polynomials $f$, $g$ are \emph{relatively prime} if $\LCM(\LM(f), \LM(g)) = \LM(f) \cdot \LM(g)$. Given a finite set $G \subseteq \Field[x_1,\ldots,x_n]$, suppose that we have $f, g \in G$ such that the leading monomials of $f$
and $g$ are relatively prime. Then $S(f, g)|_G= 0$.
\end{proposition}
Note that Buchberger's algorithm is non-deterministic and the resulting \GB basis in not uniquely determined by the input. This is because the normal form $S(f,g)|_G$ (see \cite[Algorithm 1, line 8]{mastrolilli_talg21}) is not unique as already remarked.
We observe that one simple way to obtain a deterministic algorithm (see \cite{Cox}, Theorem~2, p. 91) is to replace $h:=S(f,g)|_G$ in \cite[Algorithm 1, line 8]{mastrolilli_talg21} with $h:=\reduce {S(f,g)} {G''}$ (see Definition~\ref{def:pdiv}), where in the latter $G''$ is an ordered tuple. However, there are simple cases where the combinatorial growth of the set $G$ in Buchberger's algorithm is out of control very soon.

\subsection{The ideal membership problem \texorpdfstring{$\IMP(\Gamma)$}{IMP(Gamma)}}\label{sec:IMPdef}
For a given CSP($\Gamma$)-instance $\mathcal{C}$, the \textbf{\emph{combinatorial ideal}} $\I_\Cc$ is defined as the vanishing ideal of the set $Sol(\Cc)$, i.e.~$\I_\Cc=\Ideal{Sol(\Cc)}$ (see~\cref{def:ideal}). 

\begin{definition}\label{def:IMP}
 The {\emph{\textsc{Ideal Membership Problem}}} associated with language $\Gamma$ is the problem $\IMP(\Gamma)$ in which
 the input consists of a polynomial $f\in \mathbb{Q}[X]$ and a $\CSP(\Gamma)$ instance $\Cc=(X,D,C)$ where $D\subset \mathbb{Q}$. The goal is to decide whether $f$ lies in the combinatorial ideal~$\I_\Cc$. We use $\IMP_d(\Gamma)$ to denote $\IMP(\Gamma)$ when the input polynomial $f$ has degree at most $d$.
\end{definition}

We restrict to polynomials of maximum degree $d$ and define a $d$-truncated \GB Basis.
\begin{definition}\label{def:dTruncated GB}
  If $G$ is a \GB basis of an ideal in $\mathbb{F}[x_1,\dots,x_n]$, the \textbf{d-truncated \GB basis} $G'$ of $G$ is defined as
  \begin{equation}\label{eq:Gd}
      G' := G \cap \mathbb{F}[x_1,\dots,x_n]_d,
  \end{equation}
  where $\mathbb{F}[x_1,\dots,x_n]_d$ is the set of polynomials of degree less than or equal to $d$.
\end{definition}
%
It is not necessary to compute  a \GB basis of $\I_\Cc$ in its entirety to solve the $\IMP_d$ and to find a proof of membership. Indeed, it is sufficient to compute a $d$-truncated \GB basis $G'$ with respect to a \grlex order. In fact, for a polynomial $f_0$ of degree $d$, the only polynomials of $G$ that can divide $f_0$ are those of $G'$. Moreover, the residues of such divisions have degree at most~$d$. The latter is a property of the \grlex order (see also \cref{rm:order}). 
%
%
By \cref{th:gbprop,th:imp}, the membership test can be computed by using only polynomials from $G'$ and therefore we have
\begin{equation*}
    f \in\I_\Cc \cap \mathbb{F}[x_1,\dots,x_n]_d \iff \reduce{f}{G'}=0.
\end{equation*}
From the previous observations it follows that if we can compute $G'$ in $n^{O(d)}$ time then this yields an algorithm that runs in polynomial time for any fixed $d$ to solve $\IMP_d$. 

\begin{remark}\label{rm:order}
    It is worth noting that other total degree orderings, i.e.~ordered according to the total degree first, could be used instead of \grlex order. However other ordering like \lex order alone would not be sufficient to guarantee the solvability of $\IMP_d$ in polynomial time. Indeed, by dividing a polynomial of degree $d$ by a $\GB$ basis in \lex ordering we are no longer guaranteed that the remainder is of degree at most $d$, as was the case with the \grlex order.
\end{remark}

\section{Boolean Minority}\label{sect:minority}
In this section we focus on instances $\Cc=(X,D,C)$ of $\CSP(\Gamma)$ where $\Gamma$ is a language that is closed under the
minority polymorphism and provide the proof of  \cref{thm:MainTheorem}.
To this end, let us first introduce the definition of minority operation. 
\begin{definition}\label{def:minority}
For a finite domain $D$, a ternary operation $\otimes$ is called a minority operation if  $\otimes(a,a,b)=\otimes(a,b,a)=\otimes(b,a,a)=b$ for all $a,b\in D$.  
\end{definition}
A minority polymorphism of $\Gamma$ is a minority operation which preserves $\Gamma$ (see \cref{def:polymorph}).
Note that there is only one such polymorphism for relations over the Boolean domain, and we concern ourselves only with the minority polymorphism over the Boolean domain in this paper. 

\paragraph{Overiew of the proof of \cref{thm:MainTheorem}.}
%
A high level description of the proof structure is as follows. Each constraint that is closed under the minority polymorphism can be written in terms of linear equation (mod 2) (see e.g. \cite{Chen09}).
We first express these equations in their reduced row echelon form: that is to say the `leading variable' (the variable that comes first in the lexicographic order or \lex in short, see \cref{def:lex and grlex}) in each equation does not appear in any other (mod 2) equation. 
We then show how each polynomial in (mod 2) translates to a polynomial in regular arithmetic with exactly the same 0/1 solutions. The use of elementary symmetric polynomials allows for an efficient computation of the polynomials in regular arithmetic. 
Using these, we produce a set of polynomials $G_1$ and prove that $G_1$ is the reduced \GB basis of $\I_\Cc$ in the \lex order (see \cref{lemma:Groebnerlex}). These polynomials are unwieldy to use as is for the $\IMP_d$ as these polynomials could be of degree $n$.
We provide a conversion algorithm in \cref{sec:conversion} which converts $G_1$ to the $d$-truncated reduced \GB basis $G_2$ of $\I_\Cc$ in the \grlex order.
\cref{thm:correctness_algorithm} proves the correctness and polynomial running time of the conversion algorithm. This gives the proof of \cref{thm:MainTheorem}.
%

\subsection{\GB bases in \lex order}\label{sec:GBlex}

Consider an instance $\mathcal{C}=(X=\{x_1,\dots,x_n\},D=\{0,1\},C)$ of CSP($\Gamma$) where $\Gamma$ is $\otimes$-closed. Any constraint of $\Cc$ can be written as a system of linear equations over $\textrm{GF}(2)$ (see e.g. \cite{Chen09}). 
These linear systems with variables $x_1,\dots,x_n$ can be solved by Gaussian elimination. If there is no solution, then we have from Hilbert's Weak Nullstellensatz (see e.g.~\cite{Cox,mastrolilli_talg21}) that $1\in\I_\Cc \iff Sol(\Cc)=\emptyset \iff\I_\Cc = \mathbb{Q}[X]$. If $1\in\I_\Cc$ the reduced \GB basis is $\{1\}$. We proceed only if $Sol(\Cc)\neq\emptyset$.
In this section, we assume the \lex order $>_\lex$ with $x_1>_\lex x_2>_\lex \dots >_\lex x_n$. We also assume that the linear system has $r\leq n$ equations and is already in its reduced row echelon form with $x_i$ as the leading monomial of the $i$-th equation. Let $Supp_i\subset [n]$ such that $\{x_j:j\in Supp_i\}$ is the set of variables appearing in the $i$-th equation of the linear system except for $x_i$. Let the $i$-th equation be $R_i = 0 \Mod{2}$ where 
\begin{equation}\label{eq:Ri}
    R_i := x_i\oplus f_i,     
\end{equation}
with $i\in[r]$ and $f_i$ is the Boolean function of the form $(\bigoplus_{j\in Supp_i} x_j)\oplus \alpha_i$ and $\alpha_i \in \{0,1\}$.

In the following we transform $R_i$'s into polynomials in regular arithmetic. 
The idea is to map $R_i$ to a polynomial $R_i'$ over $\mathbb{Q}[X]$ such that $a\in \{0,1\}^n$ satisfies $R_i=0$ if and only if $a$ satisfies $R_i' = 0$. Moreover, $R_i$ is such that it has the same leading term as $R_i'$. 
We produce a set of polynomials $G_1$ and prove that $G_1$ is the reduced \GB basis of $\I_\Cc$ over $\mathbb{Q}[X]$ in the \lex ordering.
We define $R'_i$ as
\begin{equation}
    R_{i}':= x_i - M(f_i) \label{eq:R'_i}
\end{equation}
where
\begin{equation}\label{eq:mod2regexpansion}
    M(f_i) = 
    \begin{cases}
    \sum\limits_{k=1}^{|Supp_i|} \left( (-1)^{k-1}\cdot 2^{k-1}\sum\limits_{\{x_{j_1},\ldots, x_{j_k}\}\subseteq Supp_i}x_{j_1}x_{j_2}\cdots x_{j_k}\right) \textrm{ when } \alpha_i=0\\
    1+\sum\limits_{k=1}^{|Supp_i|} \left( (-1)^{k}\cdot 2^{k-1}\sum\limits_{\{x_{j_1},\ldots,x_{j_k}\}\subseteq Supp_i}x_{j_1}x_{j_2}\cdots x_{j_k}\right)  \textrm{ when } \alpha_i=1
    \end{cases}
\end{equation}

 \begin{example}\label{ex:min1}
 Consider a system with just one equation with $R_1:= x_1\oplus x_2 \oplus x_3 = 0$, where  $x_1>_\lex x_2>_\lex x_3$. Then $f_1:= x_2 \oplus x_3$ and $M(f_1):= x_2 + x_3 -2x_2x_3$. The polynomial corresponding to \cref{eq:R'_i} is
\begin{equation*}
     R_{1}':= x_1-x_2-x_3+2x_2x_3.
 \end{equation*}
 The equations $R_1=0$ and $R_{1}'=0$ have the same set of 0/1 solutions and  $\LM(R_1)=\LM(R_{1}')=x_1$. 
\end{example}

\begin{lemma}\label{lemma:Groebnerlex} 
Consider the following set of polynomials:
\begin{equation}\label{eq:GrobBasisLex}
    G_1=\{R_1',\ldots, R_r',x_{r+1}^2-x_{r+1},\ldots, x_n^2-x_n\},
\end{equation}
where $R_i'$ is from \cref{eq:R'_i}. $G_1$ is the reduced \GB basis of $\I_\Cc$ in the lexicographic order $x_1>_\lex x_2>_\lex \dots>_\lex x_n$.
\end{lemma}
\begin{proof} For any two Boolean variables $x$ and $y$,
\begin{equation}\label{eq:mod2arith}
    x\oplus y = x+y- 2xy.
\end{equation}
By repeatedly using \cref{eq:mod2arith} to obtain the equivalent expression for $f_i$, we see that $R_i=0 \Mod{2}$ and $R_i'=0$ have the same set of 0/1 solutions. Therefore $\Variety{\GIdeal{G_1}}$ is equal to $Sol(\Cc)$. This implies that $\GIdeal{G_1} \subseteq\I_\Cc$. Moreover, $\LM(R_i)=\LM(R_i')=x_i$, by construction. 

Now, we first proceed to show that $G_1$ is a \GB basis for $\GIdeal{G_1}$. With this aim we consider Buchberger's Criterion (see \cref{th:crit}) which suggest us to take into consideration all the pairs of polynomials in $G_1$ and compute the $S$-polynomial.
For every pair of polynomials in $G_1$ the reduced $S$-polynomial is zero as the leading monomials of any two polynomials in $G_1$ are relatively prime (see \cref{th:rel_prime}). By Buchberger's Criterion (see \cref{th:crit}) it follows that $G_1$ is a \GB basis of $\GIdeal{G_1}$ over $\mathbb{Q}[X]$ (according to the \lex order). In fact, inspection can reveal that $G_1$ is the \textit{reduced} \GB basis of $\GIdeal{G_1}$.

It remains to show that $\I_\Cc = \GIdeal{G_1}$. Note that by construction $\GIdeal{G_1}\subseteq\I_\Cc$. Therefore, to prove that $\I_\Cc = \GIdeal{G_1}$, it is sufficient to prove that for any $p \in\I_\Cc$ we have $p\in \GIdeal{G_1}$. 
To this end, consider a polynomial $p \in\I_\Cc$ (then $p(s)=0$ for all $s\in Sol(\Cc)$). The claim follows by proving that $p|_{G_1} = 0$ as this implies $p\in \GIdeal{G_1}$. 

We start observing that $p|_{G_1}$ cannot contain the variable $x_i$ for all $1\leq i \leq r$. Hence $p|_{G_1}$ is a polynomial multilinear in $x_{r+1},x_{r+2},\dots,x_n$.  In addition, $p|_{G_1} \in\I_\Cc$ since $G_1 \subseteq\I_\Cc$.
Now, it is important to emphasize that \emph{any} tuple $\mathbf{a}=(a_{r+1},\ldots, a_n)\in D^{n-r}$ extends to a unique feasible solution in $Sol(\Cc)$. Indeed for any partial assignment $x_{r+1}:=a_{r+1},\ldots, x_{n}:=a_n$ for the variables in $\{x_{r+1},x_{r+2},\dots,x_n\}$, we can compute the corresponding unique values for the variables in $\{x_i\mid 1\leq i\leq r\}$ that satisfy $R'_1,\ldots R'_r$ (see \cref{eq:Ri} and \cref{eq:R'_i}).
Therefore, all tuples in $D^{n-r}$ are zeros of $p|_{G_1}$. The latter implies that $p|_{G_1}$ is the zero polynomial. Hence $G_1$ is the reduced \GB basis of $\I_\Cc$.
%
\end{proof}

\begin{example}
    We refer to \cref{ex:min1}, which we complement with the following simple example for \cref{lemma:Groebnerlex}. For every pair of polynomials in $G_1=\{R_{1}',x_2^2-x_2,x_3^2-x_3\}$ the reduced $S$-polynomial is zero. By Buchberger's Criterion (see e.g. \cite{Cox} or \cref{th:crit}) it follows that $G_1$ is a \GB basis over $\Real[x_1,x_2,x_3]$ (according to the specified \lex order). 
\end{example}
Note that the reduced \GB basis $G_1$ in \cref{eq:GrobBasisLex} can be efficiently computed by exploiting the high degree of symmetry in each $M(f_i)$ and using elementary symmetric polynomials with variables from $Supp_i$.
Observe that $G_1$ is the reduced Groebner basis of $\I_\Cc$. Since $\I_\Cc$ is radical then $\langle G_1\rangle=I_C$ and therefore the ideal generated by $G_1$ is radical.

\subsection{Computing a truncated \GB basis}\label{sec:conversion}

Now that we have the reduced \GB basis in \lex order, we show how to obtain the $d$-truncated reduced \GB basis in \grlex order in polynomial time for any fixed $d=O(1)$. Before we describe our conversion algorithm, we show how to expand a product of Boolean functions. This expansion will play a crucial step in our algorithm.    
\subsubsection{Expansion of a product of Boolean functions}\label{sec:ProductOfBooleanFunctions}
In this section, we show a relation between a product of Boolean functions and (mod 2) sums of the Boolean functions, which is heavily used in our conversion algorithm in \cref{sec:OurAlgorithm}. We have already seen from \cref{eq:mod2arith} that if $f,g$ are two Boolean functions,\footnote{We earlier considered Boolean variables, but the same holds for Boolean functions.}
then 

\begin{equation*}
    2\cdot f\cdot g = f+g-(f\oplus g).
\end{equation*}

By repeatedly using the above equation, the following holds for Boolean functions $f_1,$ $f_2,\dots,f_m$:

\begin{equation}\label{eq:mod2expansion}
\begin{aligned}
f_1\cdot f_2 \cdots f_m = \frac{1}{2^{m-1}}\biggl[ &\sum_{i\in[m]} f_i - \sum_{\{i,j\}\subset [m]} (f_i\oplus f_j) + \sum_{\{i,j,k\}\subset [m]} (f_i\oplus f_j \oplus f_k) + \dots +\\ 
&(-1)^{m-1}(f_1\oplus f_2 \oplus \dots \oplus f_m)\biggr].
\end{aligned}
\end{equation}
We call each Boolean function of the form $(f_{i_1}\oplus\cdots\oplus f_{i_k})$ a \textbf{Boolean term}. We call the Boolean term $(f_1\oplus f_2 \oplus \dots \oplus f_m)$ the \textbf{longest Boolean term} of the expansion.
Thus, a product of Boolean functions can be expressed as a linear combination of Boolean terms. Note that \cref{eq:mod2expansion} is \textit{symmetric} with respect to $f_1,f_2,\dots,f_m$ as any $f_i$ interchanged with $f_j$ produces the same expression. It is no coincidence that we chose the letter $f$ in the above equation: we later apply this identity using $f_j$ from $R_j:=x_j\oplus f_j$ (see \cref{sec:GBlex}). When we use \cref{eq:mod2expansion} in the conversion algorithm, we will have to evaluate a product of at most $d$ functions, i.e.~$m\leq d=O(1)$.
We now see in the right hand side of \cref{eq:mod2expansion} that the coefficient $1/2^{m-1}$ is of constant size and there are $O(1)$ many Boolean terms.

\subsubsection{Our conversion algorithm}\label{sec:OurAlgorithm}
The FGLM \cite{FAUGERE1993329} conversion algorithm is well known in computer algebra for converting a given reduced \GB basis of a zero dimensional ideal in some ordering to the reduced \GB basis in any other ordering. However, it does so with $O(nD(\langle G_1 \rangle)^3)$ many arithmetic operations, where $D(\GIdeal{G_1})$ is the dimension of the $\mathbb{R}$-vector space $\mathbb{R}[x_1,\dots,x_n]/\GIdeal{G_1}$ (see Proposition 4.1 in \cite{FAUGERE1993329}). $D(\GIdeal{G_1})$ is also equal to the number of common zeros (with multiplicity) of the polynomials from $\GIdeal{G_1}$, which would imply that for the combinatorial ideals considered in this paper, $D(\GIdeal{G_1})=\Omega(2^{n-r})$. This exponential running time is avoided in our conversion algorithm, which is a variant the FGLM algorithm, by exploiting the symmetries in \cref{eq:mod2regexpansion} and by truncating the computation up to degree $d$. 

Some notations necessary for the algorithm are as follows: $G_1$ and $G_2$ are the reduced \GB basis of $\GIdeal{G_1}$ in \lex and \grlex ordering respectively. $\LM(G_i)$ is the set of leading monomials of polynomials in $G_i$ for $i\in\{1,2\}$. Since we know $G_1$, we know $\LM(G_1)$, whereas $G_2$ and $\LM(G_2)$ are constructed by the algorithm. $B(G_1)$ is the  set of monomials that cannot be divided (considering the \lex order) by any monomial of $\LM(G_1)$. Therefore, $B(G_1)$ is the set of all multilinear monomials in variables $x_{r+1},\dots,x_n$. Similarly, $B(G_2)$ is the set of monomials that cannot by divided (considering the \grlex order) by any monomial of $\LM(G_2)$. 
Recall the definition of $f_i$ for $i\leq r$ from \cref{sec:GBlex}. For notational purposes, we define the Boolean function $f_i:=x_i$ for $i>r$.

\begin{lemma}\label{lem:q in Boolean Terms}
Consider a monomial $q=x_{i_1}x_{i_2}\cdots x_{i_k}$ where $k\leq d$. Then $\reduce{q}{G_1}$ can be expressed as a linear combination of Boolean terms.
\end{lemma}
\begin{proof}
From \cref{eq:Ri,eq:R'_i}, $\reduce{q}{G_1} = f_{i_1}f_{i_2}\cdots f_{i_k}$ and the lemma holds using \cref{eq:mod2expansion}.
\end{proof}
Let elements $b_i$ of $B(G_2)$ be arranged in increasing \grlex order. We construct a set $A$ in our algorithm such that its elements $a_i$ are defined as $a_i=\reduce{b_i}{G_1}$ written as linear combinations of Boolean terms using \cref{lem:q in Boolean Terms}. 
We say that a Boolean term $f$ of $a_i$ ``appears in $a_j$'' for some $j<i$ if the longest Boolean term of $a_j$ is $f\oplus\alpha$ where $\alpha=0/1$. 

Let $Q$ be the set of all monomials $m$ such that $1<_\grlex deg(m) \leq_\grlex d$. We recommend the reader to refer to the example in \cref{sec:example}
for an intuitive working of the algorithm. The conversion is described in full in \cref{alg:conversion} (we assume $1 \notin\I_\Cc$, else $G_1=\{1\}=G_2$ and we are done).

\begin{algorithm}[ht]
\DontPrintSemicolon
\SetKwInput{KwInput}{Input}              
\SetKwInput{KwOutput}{Output} 
\SetKwInput{KwInit}{Initialization} 

\KwInput{Degree $d$, $G_1$, $Q$.}
\KwOutput{$d$-Truncated versions of $G_2$, $B(G_2)$.}
\KwInit{$G_2= \emptyset$, $B(G_2)=A=\{1\}$, $a_1=b_1=1$.}

\While{$Q\neq \emptyset$ }{
Let $q$ be the smallest (according to \grlex order) monomial in $Q$.\;
    Find ${q}|_{G_1}$.\; 
    Expand ${q}|_{G_1}$ using \cref{eq:mod2expansion}. 
    
    \If{the longest Boolean term of ${q}|_{G_1}$ does not appear in any $a\in A$}{
    Write ${q}|_{G_1}$ as a linear combination of $\reduce{b_i}{G_1}$ and its longest Boolean term (see \cref{lem:C}).\; Add this polynomial to $A$ and add $q$ to $B(G_2)$.\; 
    }
         
     \Else{
     Every Boolean term of ${q}|_{G_1}$ can be written as linear combinations of $\reduce{b_j}{G_1}$'s. Note that if the longest Boolean term $f$ appears in $a$ as $f\oplus 1$, then we use $f\oplus 1 = 1- f$ (see \cref{eq:mod2arith}). Thus we have ${q}|_{G_1}=\sum_j k_j\reduce{b_j}{G_1} \implies q-\sum_j k_jb_j \in \GIdeal{G_1}$.\; 
     Add the polynomial $q-\sum_j k_jb_j$ to $G_2$ and $q$ to $\LM(G_2)$. \;
     Delete any monomial in $Q$ that $q$ can divide.
     }   

     Delete $q$ from $Q$. 
}
$G_2$ is the $d$-truncated reduced \GB basis.
\caption{Computing the $d$-truncated reduced \GB basis}\label{alg:conversion}
\end{algorithm}

\begin{lemma}\label{lem:C}
The set $A$ is such that every $a_i$ is a linear combination of existing $\reduce{b_j}{G_1}$'s ($j<i$) and the longest Boolean term of $\reduce{b_i}{G_1}$.
\end{lemma}
\begin{proof}
By definition, element $a_i$ is added to $A$ when a monomial $q$ is added to $B(G_2)$ where $b_i=q$ and $a_i=\reduce{b_i}{G_1}$ expressed in Boolean terms (see \cref{alg:conversion}). 
This means that $q$ is not divisible by any monomial in $\LM(G_2)$. We prove the lemma by induction on the degree of $q$. Note that $b_1 = 1$ and hence $a_1 = \reduce{b_1}{G_1} = 1$.

If $deg(q)=1$, then $q$ is some $x_i$ and $\reduce{x_i}{G_1}$ is one of $0,1$ or $f_i$. If $\reduce{x_i}{G_1}$ is either 0 or 1, then it appears in $a_1$. We are now in the `else' condition of \cref{alg:conversion}, so $q$ should be added to $\LM(G_2)$ and not $B(G_2)$. Hence $\reduce{x_i}{G_1}$ can be neither 0 nor 1 and the lemma holds for $deg(q)=1$ as $f_i$ is the longest Boolean term.

Let us assume the statement holds true for all monomials with degree less than $m$. Consider $q$ such that $deg(q)=m$ and $q=x_{i_1}x_{i_2}\dots x_{i_m}$ where $i_j$'s need not be distinct, and the lemma holds for every monomial $<_{\grlex} q$. Then $\reduce{q}{G_1} = f_{i_1}\cdot f_{i_2}\cdots f_{i_m}$.
Let $(f_{j_1}\oplus \cdots \oplus f_{j_k})$ be a Boolean term in the expansion of $\reduce{q}{G_1}$ (by using \cref{eq:mod2expansion}), that is not the longest Boolean term, so $\{j_1,\dots,j_k\}\subset \{i_1,\dots,i_m\}$ and $k<m$. 
Consider the monomial $x_{j_1}x_{j_2}\dots x_{j_k}$. We will now prove that $x_{j_1}x_{j_2}\dots x_{j_k}$ is in fact some $b_l\in B(G_2)$ and there exists $a_l\in A$ which is a linear combination of $\reduce{b_i}{G_1}$'s and $(f_{j_1}\oplus \cdots \oplus f_{j_k})$.
The monomial $x_{j_1}x_{j_2}\dots x_{j_k}$ either belongs to $LM(G_2)$ or $B(G_2)$.
If $x_{j_1}x_{j_2}\dots x_{j_k}\in\LM(G_2)$ then it divides $q$, a contradiction to our choice of $q$. Therefore, $x_{j_1}x_{j_2}\dots x_{j_k}=b_l\in B(G_2)$.
Clearly $b_l<_{\grlex}q$ and the induction hypothesis applies, so there exists $a_l\in A$ such that

\begin{equation*}
    \reduce{b_l}{G_1}=a_l=\sum_{i < l} c_i\reduce{b_i}{G_1} + c_0(f_{j_1}\oplus \cdots \oplus f_{j_k})
\end{equation*}

where $c_i$'s are constants.
Then we simply use the above equation to substitute for the Boolean term $f_{j_1}\oplus \cdots \oplus f_{j_k}$ in $\reduce{q}{G_1}$ as a linear combination of $\reduce{b_i}{G_1}$ where $i\leq l$. We can do this for every Boolean term of $\reduce{q}{G_1}$ except the longest one. Hence, the lemma holds.
\end{proof}

\begin{theorem}\label{thm:correctness_algorithm}
Let $\Gamma$ be a constraint language over the Boolean domain that is closed under the minority polymorphism. Consider any instance $\mathcal{C}=(X,D=\{0,1\},C)$ of $\CSP(\Gamma)$ and let $G_1$ be the reduced \GB basis of $\I_\Cc$ in lexicographic order.
The conversion algorithm terminates for every input $G_1$ and correctly computes a $d$-truncated reduced \GB basis, with the \grlex ordering, of the ideal $\GIdeal{G_1}$ in polynomial time.
\end{theorem}
\begin{proof}
\cref{alg:conversion} runs at most $|Q| = O(n^d)$ times. Evaluation of any $\reduce{q}{G_1}$ can be done in $O(n)$ steps (see \cref{eq:mod2expansion}), checking if previous $a_i$'s appear (and replacing every Boolean term appropriately if it does) takes at most $O(n^d)$ steps since there are at most $|Q|$ many elements in $A$. Hence, the running time of the algorithm is $O(n^{2d})$.

Suppose the set of polynomials $\{g_1, g_2, \dots, g_k\}$ is the output of the algorithm for some input $G_1$. Clearly, $deg(g_i)\leq d$ for all $i\in [k]$. We now prove by contradiction that the output is the $d$-truncated \GB basis of the ideal $\GIdeal{G_1}$ with the \grlex ordering.  Suppose $g$ is a polynomial of the ideal with $deg(g)\leq d$, but no $\LM(g_i)$ can divide $\LM(g)$. In fact, since every $g_i\in \GIdeal{G_1}$ we can replace $g$ by $\reduce{g}{\{g_1, g_2, \dots, g_k\}}$ ($g$ generalises the reduced $S$-polynomial). The fact that $g\in\GIdeal{G_1}$ and $\reduce{g}{G_1}=0$ implies that $\LM(g)$ is a linear combination of monomials that are less than $\LM(g)$ (in the \grlex order) and hence must be in $B(G_2)$, i.e. 
\begin{equation*}
    \reduce{g}{G_1}=0 \implies \reduce{\LM(g)}{G_1}=\sum_i k_i\reduce{b_i}{G_1},
\end{equation*}
where every $b_i\in B(G_2)$ and $b_i<_\grlex \LM(g)$. When the algorithm runs for $q=\LM(g)$, since $q$ was not added to $\LM(G_2)$,
\begin{equation*}
    \reduce{\LM(g)}{G_1}= \sum_j k_j\reduce{b_j}{G_1} + f,
\end{equation*}
where $f$ is the longest Boolean term of $\reduce{\LM(g)}{G_1}$ which does not appear in any previous element of $A$. But the two equations above imply that $\sum_i k_i\reduce{b_i}{G_1} = \sum_j k_j\reduce{b_j}{G_1} + f$, which proves that there exists some $b_l\in B(G_2)$ such that $a_l$ has $f$ as its longest Boolean term, so $f$ should have appeared in $a_l$, a contradiction. Therefore the output is a $d$-truncated \GB basis. Although unnecessary for the $\IMP_d$, we also prove that the output is reduced: every non leading monomial of every polynomial in the output comes from $B(G_2)$ and no leading monomial is a multiple of another by construction.
\end{proof}
Thus we have a proof of \cref{thm:MainTheorem} and \cref{thm:Corollary_minority}. 

\section{Dual discriminator}\label{sect:ddiscriminator}
In this section we focus on instances $\Cc=(X,D,C)$ of $\CSP(\Gamma)$ where $\Gamma$ is a language that is closed under the
dual discriminator polymorphism and provide the proof of~\cref{thm:main_dualdisc}.
To this end, let us first provide some definitions and preliminaries. 

We assume $D\subset \mathbb{Q}$ is any finite domain and the polymorphism in question is the dual discriminator operation $\nabla$ that preserves $\Gamma$ (see \cref{def:polymorph}). The dual discriminator is a majority operation \cite{Jeavons:1997:CPC,barto_et_al:DFU:2017:6959} and is often used as a starting point in many CSP-related classifications \cite{barto_et_al:DFU:2017:6959}.
For a finite domain $D$, a ternary operation $f$ is called a majority operation if  $f(a,a,b)=f(a,b,a)=f(b,a,a)=a$ for all $a,b\in D$. 
\begin{definition}\label{def:dual discriminator}
  The dual discriminator on a domain $D$, denoted by $\nabla$, is a majority operation such that $\nabla(a,b,c)=a$ for pairwise distinct $a,b,c\in D$.
\end{definition}
%
%
Observe that both $R_1$ and $R_2$ from \cref{example:CSP} are $\nabla$-closed.
The well-known 2-SAT problem is another example of a CSP whose language is closed under the Boolean majority polymorphism (note that the majority operation on a Boolean set is also a dual discriminator operation).

In the remainder we will assume that the solution set is non-empty, as this can be verified in polynomial time using the Arc-Consistency (AC) algorithm \cite{MACKWORTH1977}.

The constraints for $\nabla$-closed problems can be assumed to be binary \cite{Jeavons:1997:CPC,Pixley1975} and are of three types: permutation constraints, complete constraints and two-fan constraints \cite{szendrei,Jeavons1994_Characterising_tract_constraints}. Bulatov and Rafiey \cite{bulatov2020complexity} proved that the $\IMP_d(\Gamma)$ over a finite domain is decidable in polynomial time, without showing a proof of membership. They did so by cleverly eliminating the permutation constraints, but were unable to recover a proof for the original problem. We show that a \GB basis of the ideal restricted to the permutation constraints can be computed in polynomial time in \cref{sec:permconstraints}. We then show in \cref{sec:completeandtwofan} that the \GB basis of constraints that are complete and two-fan constraints can come from a fixed set (see \cref{def:DFL}). 
We prove in \cref{sec:GB of entire ideal} that the \GB basis of the entire ideal can be found in polynomial time. This \GB basis is independent of degree $d$ of the input polynomial: it only contains polynomials with degree less than or equal to $|D|!$. 

\subsection{Permutation constraints}\label{sec:permconstraints}
For $i,j\in [n]$, a permutation constraint is of the form $R(x_i,x_j)$, where $R=\{ (a,\pi_{ij}(a))\mid a\in D_{ij}\}$ for some $D_{ij}\subseteq D$ and some bijection $\pi_{ij}:D_{ij}\to D_{ij}'$, where $D_{ij}'\subseteq D$. Clearly $|D_{ij}|=|D_{ij}'|$. Let $\ca{P}$ be the set of input permutation constraints.
We can assume that there exists at most one permutation constraint over every pair of variables: if there are two on the same set of variables, then their intersection is a permutation constraint. 
Let $R_{ij}(x_i,x_j)$ represent the unique permutation constraint on variables $x_i,x_j$, if one exists.

Informally, the goal is to make larger constraints called \textit{chained permutation constraints} (CPC's). 
Permutation constraints on overlapping variables can be linked to form a larger constraint by using bijections. 
For example, if there exists $R_{ij}(x_i,x_j),R_{jk}(x_j,x_k)\in\ca{P}$, we form a new constraint on $x_i,x_j,x_k$ by using $\pi_{ij}$ and $\pi_{jk}$: the chained permutation constraint is $R(x_i,x_j,x_k)$ where 

\begin{equation*}
    R=\{(\pi_{ij}^{-1}(a),a,\pi_{jk}(a))\mid a\in D'_{ij}\cap D_{jk}\}.
\end{equation*}

The number of tuples in any CPC is always less than or equal to the domain size, since there is always a bijection between any two variables of a CPC. On the other side the arity of the chained permutation constraint can be as large as $n$.
The constructing of CPC's can be carried out by the arc consistency algorithm described in \cite{MACKWORTH1977} and 
\cref{alg:BuildingCPC} is tailored to our application.
A brief working of \cref{alg:BuildingCPC} is as follows: let $J\subset [n]$ be an index set for the CPC's (it becomes clear later why there can be at most $\lfloor n/2 \rfloor$ of them but we use $n$ for convenience). 
We initialise $J=\emptyset$. 
A general chained permutation constraint $\CPC_i$ is defined as $R_i(X_i)$ where $R_i$ is a relation and $X_i\subseteq X$ is a variable set. 
We keep track of the values that each variable is allowed to take, that is, $S_a\subseteq D$ is the set of solutions of $x_a$ that satisfies $\CPC_i$ for all $x_a\in X_{i}$. 
The sets $S_a$ and $X_{i}$ are updated as $\CPC_i$ grows. We define $\sigma_{ab}:S_a\to S_b$ to be the bijection between any two pairs of variables $x_a,x_b\in X_{i}$. Note that $|S_a|=|S_b|$ because the possible values of $x_a$ and $x_b$ are in bijection whenever they belong to some $X_i$. Hence $\sigma_{ba}=\sigma_{ab}^{-1}$. Let $\sigma_{aa}$ denote the identity function for all $a\in [n]$. For any permutation constraint $R_{pq}$ in $\ca{P}$, one of the four is true.
\begin{itemize}
    \item Neither $x_p$ nor $x_q$ belong to $\cup_{j\in J} X_j$: in which case we create a $\CPC$. We define $\CPC_i=R_{pq}(x_p,x_q)$ and $X_{i}=\{x_p,x_q\}$ where $i\in [n]\setminus J$. Update $J:=J\cup \{i\}$.
    \item $x_p\in X_i$ for some $i\in J$ and $x_q\notin \cup_{j\in J}X_{j}$, in which case we expand $\CPC_i$ to include $R_{pq}$ and $x_q$ is included in $X_i$.
    \item $x_p\in X_i$, $x_q\in X_j$ and $i\neq j$, in which case $\CPC_i$ and $\CPC_j$ have a permutation constraint linking two of their variables, so we combine the two CPC's into one. The set $X_i\cup X_j$ is the new $X_i$ and $j$ is deleted from $J$.
    \item Both $x_p,x_q\in X_{i}$, in which case we update $\CPC_i$ to retain only the common solutions between $\CPC_i$ and $R_{pq}$.
\end{itemize}
As $R_{pq}(x_p,x_q)$ is now accounted for in some $\CPC$, it is deleted from $\mathcal{P}$. The algorithm runs until $\ca{P}$ is empty. Once $\ca{P}$ is empty, for each $i\in J$, $\CPC_i$ is defined as
 \begin{equation*}
     \CPC_i:=R_i(X_{i}=\{x_{i_1},x_{i_2},\ldots,x_{i_r}\}),
 \end{equation*}
where $R_i$ is the $r$-ary relation 
  $R_i = \{(a,\sigma_{i_1i_2}(a),\ldots,\sigma_{i_1i_r}(a))\mid a\in S_{i_1}\}$.
The following lemma is straightforward because of the way each CPC is constructed.
\begin{lemma}\label{rem:X_CPCi are independent}
For $i\neq j$ and $i,j\in J$, $X_{i}\cap X_{j}=\emptyset$.
\end{lemma}
%
%
Note that the transitive closure of the permutation constraints does not change $Sol(\Cc)$, so it does not change the ideal $\I_\Cc$, since~$\I_\Cc=\Ideal{Sol(\Cc)}$ (see \cref{sec:IMPdef}).
\begin{lemma}\label{lem:GB of I_CPC_i}
Let $I_{\CPC_i}$ be the combinatorial ideal associated with $\CPC_i$. A \GB basis of $\sum_{i\in J} I_{\CPC_i}$ can be calculated in polynomial time.
\end{lemma}

\begin{proof}
Suppose that we see a relation as a matrix where each tuple is a row. Then the tuple size is the number of columns. For $i\in J$, the relation $R_i$ in $\CPC_i=R_i(X_{i})$ is such that it can have at most $|D|!$ pairwise distinct columns as there exists a bijection between every pair of variables in $X_{i}$. If the columns corresponding to $x_j$ and $x_k$ are the same, for some $x_j,x_k\in X_i$, then $x_j-x_k\in I_{\CPC_i}$. 
Consider the set 
\begin{equation*}
    \{x_j-x_k\mid x_j>x_k \textrm{ and } \sigma_{jk} \textrm{ is the identity function } \forall x_j,x_k\in X_i\}.
\end{equation*}
Let $\ca{S}_i$ be the reduced \GB basis of the ideal generated by this set. It is not difficult to see that polynomials in $\ca{S}_i$ are of the form $x_j-x_k$.
Let $M_i$ be the set of leading monomials of the polynomials in $\mathcal{S}_i$. Consider the variables in $X_{i}\setminus M_i$.
Since there is a bijection $\sigma_{jk}$ between every pair of variables $x_j,x_k\in X_{i}\setminus M_i$, we can produce a pair of interpolating polynomials $x_j-f(x_k)$ and $x_k-g(x_j)$ such that $x_j=a,x_k=\sigma_{jk}(a)$ satisfies both polynomials for every $a\in S_j$.
The degree of $f$ and $g$ is at most $|D|-1$ since $|D|$ is the domain size and there are at most $|D|$ tuples in $R_i$. Let $\mathcal{T}_i$ be this set of interpolating polynomials for every pair of variables in $X_{i}\setminus M_i$. Then \begin{equation}\label{eq:I_CPCi}
    I_{\CPC_i} = \GIdeal{\mathcal{T}_i}+ \GIdeal{\{\Pi_{a\in S_j} (x_j-a)\mid x_j\in X_{i}\setminus M_i\}}+\GIdeal{\mathcal{S}_i}.
\end{equation} 
Since $|X_{i}\setminus M_i|\leq |D|!=O(1)$ and the degree of every polynomial in \cref{eq:I_CPCi} is less than or equal to $|D|$, a \GB basis of $\GIdeal{\mathcal{T}_i}+ \GIdeal{\{\Pi_{a\in S_j} (x_j-a)\mid x_j\in X_{i}\setminus M_i\}}$ can be computed by the Buchberger's algorithm in time that is at most exponential in the number of variables, i.e.~in $O(1)$ time for any $|D|=O(1)$. 
Suppose $G_i$ is the reduced \GB basis. Then $\ca{G}_i=G_i\cup \ca{S}_i$ is the reduced \GB basis of $I_{\CPC_i}$: for every pair of polynomials $f\in G_i$ and $g\in\ca{S}_i$, $\reduce{S(f,g)}{\{f,g\}}=0$ as $\LM(f)$ and $\LM(g)$ are relatively prime (see \cref{th:rel_prime}). 
Thus the reduced \GB basis of $I_{\CPC_i}$ can be calculated in polynomial time. Since there can be at most $\lfloor n/2 \rfloor$ $\CPC_i$'s, and from \cref{rem:X_CPCi are independent}, the reduced \GB basis of $\sum_i I_{\CPC_i}$ is $\bigcup_i \mathcal{G}_i$ and can be calculated in polynomial time.
\end{proof}
We will need to find the \GB basis of $I_{\CPC_i}$ later in the proof of \cref{thm:majority_maintheorem}. We show in the \cref{sec:GB of entire ideal} as to how we can compute the rest of the polynomials in the \GB basis of $\I_\Cc$ by using the relations $R_i$, the sets $S_j$, the bijections $\sigma_{kl}$ and polynomials that define complete and two-fan constraints. 
\begin{algorithm}
\DontPrintSemicolon
\SetKwInput{KwInput}{Inputs}              
\SetKwInput{KwOutput}{Outputs} 
\SetKwInput{KwInit}{Initialization} 
  \KwInput{Set of permutation constraints $\mathcal{P}$, index set $J$, Chained permutation constraint $\CPC_i$ for every $i\in J$.}
  \KwOutput{Index set $J$, Chained permutation constraint $\CPC_i$ for every $i\in J$.}
  \While{$\mathcal{P}\neq \emptyset$}{
  Select $R_{pq}(x_p,x_q)$ from $\mathcal{P}$.\;
  \If{$\{x_p,x_q\}\not\subseteq \cup_j X_j$}{\label{line:beginInitNewCPC}
  Let $i\in [n]\setminus J$. $J:=J\cup \{i\}$.\tcc*{Create $\CPC_i$.}
  $X_{i}:=\{x_p,x_q\}$,
  $S_p:=D_{pq}$ and $S_q:=D'_{pq}$.\;
  Let $\sigma_{pq}:S_p\to S_q$ where $\sigma_{pq}:=\pi_{pq}$.\;\label{line:endInitNewCPC}
  }
  \If{$x_p\in X_i$ and $x_q\notin \cup_j X_j$}{
  $S_p := S_p\cap D_{pq}$.\label{line:if_Sp}\tcc*{Expand $\CPC_i$.}
  $S_j:=\{\sigma_{pj}(a)\mid a\in S_p\}$ for every $x_j\in X_{i}\setminus \{x_p\}$.\;\label{line:if_Sj}
  $S_q:=\{\pi_{pq}(a)\mid a\in S_p\}$.\;\label{line:if_Sq}
  $X_{i}:= X_{i}\cup\{x_q\}$.\;\label{line:if_X_cpc}
$\sigma_{pq}:= S_p \to S_q$ where $\sigma_{pq}=\pi_{pq}$.\;\label{line:if_Sigmapq}
    $\sigma_{jq}:= S_j \to S_q$ where $\sigma_{jq}=\sigma_{pq}(\sigma_{jp})$.\;\label{line:if_Sigmajq}
    }
    \If{$x_p\in X_i$, $x_q\in X_j$ and $i\neq j$}{
    $S_p := \{a\in S_p \cap D_{pq} \mid \pi_{pq}(a)\in S_q \}$.\label{lin:Spdef}\tcc*{Combine $\CPC_i$ and $\CPC_j$.}
  $S_k:=\{\sigma_{pk}(a)\mid a\in S_p\}$ for every $x_k\in X_{i}\setminus \{x_p\}$.\;\label{lin:Sk_start}
  $S_q:=\{\pi_{pq}(a)\mid a\in S_p\}$.\;
  $S_k:=\{\sigma_{qk}(a)\mid a\in S_q\}$ for every $x_k\in X_{j}\setminus \{x_q\}$.\;\label{lin:Sk_end}
  $\sigma_{pq}:= S_p\to S_q$ where $\sigma_{pq}=\pi_{pq}$ where $a\in S_p$.\;\label{lin:sigmakl_start}
  $\sigma_{kl}:= S_k \to S_l$ $\forall x_k\in X_i\setminus\{x_p\},x_l\in X_j\setminus\{x_q\}$ where $\sigma_{kl}=\sigma_{ql}(\sigma_{pq}(\sigma_{kp}))$.\;\label{lin:sigmakl_end}
  $X_i:=X_i\cup X_j$. $J:=J\setminus \{j\}$.\label{lin:combineX}
    }
    \If{$x_p,x_q\in X_i$}{
    $S_p := \{a\in S_p \cap D_{pq} \mid \sigma_{pq}(a)=\pi_{pq}(a) \}$.\label{line:else_Sp}\tcc*{Update $\CPC_i$.}
  $S_j:=\{\sigma_{pj}(a)\mid a\in S_p\}$ for every $x_j\in X_{i}\setminus \{x_p\}$.\;\label{line:else_Sj}
    }
    Delete $R_{pq}(x_p,x_q)$ from $\mathcal{P}$.\;
  }
  For each $i\in J$, we define
  $\CPC_i:=R_i(X_{i}=\{x_{i_1},x_{i_2},\ldots,x_{i_r}\})$ where
  $R_i = \{(a,\sigma_{i_1i_2}(a),\ldots,\sigma_{i_1i_r}(a))\mid a\in S_{i_1}\}$.\label{line:CPCi}
\caption{Constructing chained permutation constraints}\label{alg:BuildingCPC}
\end{algorithm}

\subsection{Complete and two-fan constraints}\label{sec:completeandtwofan}
A complete constraint is of the form $R(x_i,x_j)$ where $R=D_i\times D_j$  for some $D_i,D_j \subseteq D$. The polynomials that can represent these constraints are $\prod_{a\in D_i} (x_i-a)$ and $\prod_{b\in D_j} (x_j-b)$. We call these polynomials \textit{partial domain polynomials}. If no such input explicitly exists for a variable, the domain polynomial in that variable itself is the partial domain polynomial. A two-fan constraint is of the form $R(x_i,x_j)$ where $R=\{(\{a\}\times D_j)\cup (D_i\times \{b\})\}$ for some $D_i,D_j \subseteq D$ with $a\in D_i,b\in D_j$. This constraint can be represented by the set of polynomials  $\{(x_i-a)(x_j-b),\prod_{c\in D_i} (x_i-c), \prod_{d\in D_j} (x_j-d)\}$.

\begin{definition}\label{def:DFL}
 The set of polynomials $\ca{D}$, $\ca{F}$ and $\ca{L}$ is defined as follows: 
 \begin{align*}
    \mathcal{D} &= \{\Pi_{a\in A} (x_i-a)\mid i\in [n], A\subseteq D\},\\
    \mathcal{F} &= \{(x_i-a)(x_j-b)\mid i,j\in [n],i\neq j\},\\
    \mathcal{L} &=  \{x_i-\alpha_2 - (x_j-\beta_2)(\alpha_1-\alpha_2)/(\beta_1-\beta_2) \mid i,j\in [n],i\neq j \},
\end{align*}
for all $a,b,\alpha_1,\alpha_2,\beta_1,\beta_2\in D$ where $\alpha_1\neq \alpha_2$ and $\beta_1\neq \beta_2$.
\end{definition}
In other words $\ca{D}\cup\ca{F}$ is the set of all complete constraints and two-fan constraints and $\ca{L}$ is the set of polynomials in two variables that pass through two points $(\alpha_1,\beta_1),(\alpha_2,\beta_2)\in D^2$ where $\alpha_1\neq \alpha_2$ and $\beta_1\neq \beta_2$. Let $G\subset \ca{D}\cup\ca{F}$ be the set of polynomials that describes the input complete constraints and two-fan constraints of an instance of $\IMP_d(\Gamma)$. Let $I_{\CF}=\GIdeal{G}$ (combinatorial ideal for the Complete and two-Fan constraints). 
Recall that the constraints for $\nabla$-closed problems can be assumed to be binary \cite{Jeavons:1997:CPC} and are of three types: permutation constraints, complete constraints and two-fan constraints \cite{szendrei,Jeavons1994_Characterising_tract_constraints}. Then
\begin{equation*}
   \I_\Cc = \sum_{i\in J}I_{\CPC_i} + I_\CF = \sum_{i\in J}I_{\CPC_i} + \GIdeal{G}. 
\end{equation*}

\begin{lemma}\label{lem:red GB of D union F}
The reduced \GB basis of $I_\CF$ can be calculated in polynomial time and is a subset of $\ca{D}\cup\ca{F}\cup\ca{L}$.
\end{lemma}
\begin{proof}
The proof follows by a case analysis. Some cases have already been considered in \cite[Lemma 5.16]{bulatov2020complexity} and in \cite[Section 3, Case 2]{bharathi_et_al:MFCS2020}, and in the following we will refer to those references for brevity of discussion.

For any pair $f,g\in G$, we show that there are polynomials $H \subset \GIdeal{G}$ such that $H\subset \ca{D}\cup\ca{F}\cup\ca{L}$ and $\reduce{S(f,g)}{H}=0$. We then include these polynomials in $G$, i.e.~$G:=G\cup H$. We dismiss the cases where $\LM(f)$ and $\LM(g)$ are relatively prime (see \cref{th:rel_prime}). 

The cases already considered in \cite[Lemma 5.16]{bulatov2020complexity} are when:
\begin{itemize}
    \item $f,g\in\ca{F}$ where $f = (x_i-a)(x_j-b)$, $g = (x_i-c)(x_k -d)$ for $a=c$ and $a\neq c$,
    \item $f\in\ca{D},g\in\ca{F}$ where $f = \Pi_{a\in D_i} (x_i-a)$, $g = (x_i -c)(x_j - b)$ and $c\in D_i$.
\end{itemize}

We now consider the remaining cases. When a polynomial $h\in\ca{L}$ is added to $G$, we reduce all polynomials in $G$ by $h$ and therefore no other polynomial in $G$ contains $\LM(h)$.
\begin{enumerate}
    \item Suppose $f\in\ca{D},g\in\ca{F}$ where $f=\Pi_{a\in D_i}(x_i-a)$, $g=(x_i-c)(x_j-b)$ and $c\notin D_i$. 
    Every common zero of $f$ and $g$ has to satisfy $x_j = b$ (since $x_i=c$ is not feasible).
    This implies $x_j-b \in \GIdeal{G}$. Since this polynomial can divide $g$, we replace $g$ with $x_j-b$ instead and continue with the Buchberger's algorithm. With this new $g$, $\reduce{S(f,g)}{\{f,g\}}=0$ since the leading monomials are relatively prime (see \cref{th:rel_prime}). 
    \item Suppose $f,g\in\ca{F}$, where $f=(x_i-a)(x_j-b)$ and $g=(x_i-c)(x_j-d)$. If $a\neq c$ and $b\neq d$, then $R=\{(a,d),(c,b)\}$ is the set of solutions to both $f$ and $g$.
    The constraint $R(x_i,x_j)$ is a new permutation constraint and is added to $\ca{P}$. We run \cref{alg:BuildingCPC} again to update CPC's or create one.
    A \GB basis for the ideal $\GIdeal{f,g}$ is $H=\{(x_i-a)-\frac{c-a}{b-d}(x_j-d),(x_i-a)(x_i-c),(x_j-b)(x_j-d)\}$ (see \cite[Section 3, Case 2]{bharathi_et_al:MFCS2020}) and therefore $H\subset \ca{D}\cup\ca{F}\cup\ca{L}$. 
    The $S$-polynomial $S(f,g)=(d-b)x_i+(c-a)x_j+ab-cd$, hence $\reduce{S(f,g)}{H}=0$. Note that after running \cref{alg:BuildingCPC}, $x_i,x_j\in X_k$ for some $k\in J$. So, the linear polynomial in $H$ belongs to $I_{\CPC_k}$.
    If $a=c$ and $b\neq d$, then the S-polynomial is $S(f,g)=x_i(d-b)+a(b-d)=(d-b)(x_i-a)$ and $H=\{x_i-a\}$. 
    \item Suppose $f\in \ca{D}\cup\ca{F}\cup\ca{L}$ and $g\in \ca{L}$. The polynomials $f$ and $g$ are relatively prime since if $g$ were to belong to $G$, then no other polynomial contains $\LM(g)$. This implies $\reduce{S(f,g)}{\{f,g\}}=0$.
\end{enumerate}

Hence, the $S$-polynomial for every two polynomials in $G$ is such that there are polynomials in $I_\CF$ that belong to $\ca{D}\cup\ca{F}\cup\ca{L}$ that reduce the $S$-polynomial to zero. In fact, it is not difficult to see that the reduced \GB basis is also a subset of $\ca{D}\cup\ca{F}\cup\ca{L}$. Since $|\ca{D}\cup\ca{F}\cup\ca{L}|=O(n^2)$, the reduced \GB basis of $I_\CF$ can be calculated in polynomial time.
\end{proof}
\begin{algorithm}
\DontPrintSemicolon
\SetKwInput{KwInput}{Input}              
\SetKwInput{KwOutput}{Output} 
\SetKwInput{KwInit}{Initialization} 
 
  \KwInput{$G,\CPC_i$.}
  \KwOutput{\GB basis of $\I_\Cc$.}
  Compute and replace $G$ by the reduced \GB basis of $I_\CF$.\;\label{lin:GB of AUB} 
  \For{every $g=\Pi_{a\in D_p}(x_p-a)\in G\cap\ca{D}$
  }{
  \If{$D_p\neq S_p$}{
  $S_p:= S_p\cap D_p$. Suppose $x_p\in X_i$.\;
  $S_k:= \{\sigma_{pk}(a)\mid a\in S_p \}$ for every $x_k\in X_i\setminus\{x_p\}$.\;
  Replace $g$ by $\Pi_{a\in S_p}(x_p-a)$ in $G$. 
  Go to line 1.  
  }
  }
  Let $C= G\cap\ca{F}$.\;
  \While{$C\neq\emptyset$}{
  Choose $g=(x_p-a)(x_q-b)\in C$. 
  Suppose $x_p\in X_i$.\;
  \If{$a\notin S_p$}{
  Add $x_q-b$ to $G$ if $a\notin S_q$ else add $x_p-a$ to $G$. Go to line 1.  
  }
  \tcc{At this point $a\in S_p$ and $b\in S_q$.}
  
  \If{$x_q\in X_j$ for some $i\neq j$}{
  \If{$b\notin S_q$}{
  Add $x_p-a$ to $G$. Go to line 1.  
  }
  \tcc{At this point $a\in S_p$ and $b\in S_q$.}
  Let $B:= \{(x_k-\sigma_{pk}(a))(x_l-\sigma_{ql}(b))\mid x_k\in X_i,x_l\in X_j\}\setminus\{g\}$.
  
  \If{$\exists h\in B$ such that $\reduce{h}{G}\neq 0$}{
  $G:=G\cup B$. Go to line 1.  
  }
  }
  \If{$x_q\notin \cup_{j\in J}X_j$}{
  Let $B:= \{(x_k-\sigma_{pk}(a))(x_q-b)\mid x_k\in X_i\}$.\;
  
  \If{$\exists h\in B$ such that $\reduce{h}{G}\neq 0$}{
  $G:=G\cup B$. Go to line 1.  
  }
  }
  Delete $g$ from $C$.
  }
  Calculate $\ca{G}_i$ for every $i$.\;
  A \GB basis of $\I_\Cc$ is $(\cup_i \ca{G}_i) \cup G$.
\caption{Calculating \GB basis}\label{alg:CalculatingGB}
\end{algorithm}

\subsection{Computing a \GB basis}\label{sec:GB of entire ideal}
\begin{theorem}\label{thm:majority_maintheorem}
Let $\Gamma$ be a constraint language over a finite domain that is closed under the dual discriminator polymorphism. Consider any instance $\mathcal{C}$ of $\CSP(\Gamma)$. 
A \GB basis for the combinatorial ideal $\I_\Cc$ can be calculated in polynomial time. 
\end{theorem}
\begin{proof}
Let $G$ be the reduced \GB basis of $I_\CF$. Then
\begin{equation*}
   \I_\Cc = \sum_{i\in J}I_{\CPC_i} + I_\CF = \sum_{i\in J}\GIdeal{\ca{G}_i} + \GIdeal{G}, 
\end{equation*}
where $\ca{G}_i$ is the reduced \GB basis of $I_{\CPC_i}$ (see \cref{lem:GB of I_CPC_i}). For two polynomials $f,g\in (\cup_i \ca{G}_i)\cup G$, we see what the reduced $S$-polynomial can imply. The straightforward cases are when
\begin{itemize}
    \item $f,g\in\ca{G}_i$: here $\reduce{S(f,g)}{\ca{G}_i}=0 $ since $\ca{G}_i$ is the reduced \GB basis of $I_{\CPC_i}$,
    
    \item $f\in\ca{G}_i,g\in\ca{G}_j$ where $i\neq j$: as $f$ and $g$ don't share any variable in common (see \cref{rem:X_CPCi are independent}), the leading monomials are relatively prime (see \cref{th:rel_prime}), hence $\reduce{S(f,g)}{\{f,g\}}=0$,
    
    \item $f,g\in \ca{D}\cup\ca{F}\cup\ca{L}$: here $\reduce{S(f,g)}{G}=0$ because of \cref{lem:red GB of D union F}.
\end{itemize}
The only cases to examine are when $f\in\ca{G}_i$ and $g\in G \subset \ca{D}\cup\ca{F}\cap\ca{L}$. Suppose $\LM(f)$ and $\LM(g)$ contain $x_p$ and $x_p\in X_i$.

\textbf{Case 1.} Suppose $g\in \ca{D}$ i.e., $g=\Pi_{a\in D_p} (x_p-a)$. If $D_p=S_p$, then $g\in \GIdeal{\ca{G}_i}$ and $\reduce{S(f,g)}{\ca{G}_i}=0$. If not, then we update $S_p:=S_p\cap D_p$ and $S_k:=\{\sigma_{pk}(a)\mid a\in S_p\}$ for every $x_k\in X_{i}\setminus\{x_p\}$. We replace $g$ by $\Pi_{a\in S_p} (x_p-a)$ and recalculate the reduced \GB basis of $I_\CF$. 

\textbf{Case 2.} Suppose $g\in \ca{F}$, that is, $g=(x_p-a)(x_q-b)$. If $a\notin S_p$, then $x_q=b$ is the only feasible solution to $g$. We replace $g$ by $x_q-b$ and recompute the reduced \GB basis of $I_\CF$. Considering $a\in S_p$, we now distinguish the cases based on where $x_q$ belongs. Note that $x_q\notin X_i$ else $f,g\in\GIdeal{\ca{G}_i}$. 
\begin{enumerate}
    \item $x_q\in X_{j}$ and $i\neq j$: 
    if $b\notin S_q$ then we replace $g$ by $x_p-a$ and recompute the reduced \GB basis of $I_\CF$. 
    Suppose $a\in S_p$ and $b\in S_q$. If not already present in $G$, we add polynomials $ (x_k-\sigma_{pk}(a))(x_l-\sigma_{ql}(b))$ to $G$, where $ x_k\in X_i,x_l\in X_j$. 
    We claim that $\ca{G}_i\cup\ca{G}_j\cup G$ is a \GB basis of the ideal it generates.
    A polynomial in variables $X_i\cup X_j$ that belongs to $\GIdeal{\ca{G}_i\cup\ca{G}_j\cup G}$ is such that it must have the following solutions:
    \begin{enumerate}[(i)]
        \item $x_k=\sigma_{pk}(a)$ for every $x_k\in X_i$ and variables in $X_j$ take any tuple of $R_j$, and
        \item $x_k=\sigma_{qk}(b)$ for every $x_k\in X_j$ and variables in $X_i$ take any tuple of $R_i$.
    \end{enumerate}
    Let $u=\reduce{S(f,g)}{\ca{G}_i\cup\ca{G}_j\cup G}$. We now show that $u$ is the zero polynomial. No monomial in $u$ can contain $x_kx_l$ (where $x_k\in X_{i},x_l\in X_{j}$) due to the polynomials in $G$.
    We have $u=u_1(X_{i})+u_2(X_{j})$.
    Suppose $x_p=a$ and the solutions are of type (i). Then $u_1(X_i)$ is a constant, say $c$ and $u=c + u_2(X_{j})$. Since $u\in\GIdeal{\ca{G}_i\cup\ca{G}_j\cup G}$, $u=0$ for every solution of type (i). This implies $u_2(X_j)+c=0$ for every solution in $R_j$ and hence $u_2(X_j)+c\in \GIdeal{\ca{G}_j}$ and $u_2(X_j)=-c$. Now $u=u_1(X_i)-c$. Considering solutions of type (ii), we see that $u_1(X_i)-c\in\GIdeal{\ca{G}_i}$ and hence $u_1(X_i)=c$ and thus $u=0$.
    
    \item $x_q\notin \cup_j X_j$: if not already present in $G$, we add polynomials $(x_k-\sigma_{pk}(a))(x_q-b)$ to $G$ where $x_k\in X_i$. The proof of why $\reduce{S(f,g)}{\ca{G}_i\cup G}=0$ is similar.
\end{enumerate}
\textbf{Case 3.} Suppose $g\in\ca{L}$. The polynomial $g$ represents a permutation constraint and is dealt with in the chained permutation constraint containing $x_p$, i.e.~$\CPC_i$. This implies $g\in\GIdeal{\ca{G}_i}$ (see point 2 of \cref{lem:red GB of D union F}) and $\reduce{S(f,g)}{\ca{G}_i}=0$.
\end{proof}
Clearly, this \GB basis is independent of the degree $d$ of the input polynomial. Hence, we have proof of \cref{thm:main_dualdisc,thm:corollary_dualdisc}.

\section{Conclusions}
The $\IMP_d$ tractability for combinatorial ideals has useful practical applications as it implies bounded coefficients in Sum-of-Squares proofs.
A dichotomy result between ``hard'' (NP-hard) and ``easy'' (polynomial time) $\IMP$s was achieved for the $\IMP_0$ ~\cite{Bulatov17,Zhuk17} over any finite domain nearly thirty years after that over the Boolean domain \cite{Schaefer78}. The $\IMP_d$ for $d=O(1)$ over the Boolean domain was tackled by Mastrolilli \cite{MonaldoMastrolilli2019} based on the classification of the $\IMP$ through polymorphisms, where the complexity of the $\IMP_d$ for five of six polymorphisms was solved. We solve the remaining problem, i.e.~the complexity of the $\IMP_d(\Gamma)$ when $\Gamma$ is closed under the ternary minority polymorphism. This is achieved by showing that the $d$-truncated reduced \GB basis can be computed in polynomial time, thus completing the missing link in the dichotomy result of \cite{MonaldoMastrolilli2019}. 
We also show that a proof of membership can be found in polynomial time regarding the $\IMP(\Gamma)$ for which constraints are closed under the dual discriminator polymorphism.

We believe that generalizing the dichotomy results of solvability of the IMP$_d$ for a finite domain is an interesting and challenging goal that we leave as an open problem.

As a starting point, it would be interesting to study the tractability of $\IMP_d(\Gamma)$ when $\Gamma$ is a finite set of relations over $\{0, 1, 2\}$ closed under any majority. There is actually only one majority polymorphism in the Boolean domain compared to $3^6$ majority polymorphisms in the ternary domain. The suggestion would be to start with a majority operation $m$ that satisfies the identities $m(0, 1, 2) = m(1, 2, 0) = m(2, 0, 1) = 1$ and $m(0, 2, 1) = m(2, 1, 0) = m(1, 0, 2) = 2$.

Considering working more systematically on the $3$-element case, a complete list of polymorphisms that need to be looked at, that is, if we prove that for each of them the IMP is tractable, we have a dichotomy in the 3-element case, can be find in \cite[Section 4.4.1]{brady2022notes}.
The list provided in \cite[Section 4.4.1]{brady2022notes} by Brady is the list of the (up to term-equivalence and 1 permutations of the domain) 24 minimal Taylor clones on a 3-element set.
Brady’s classification is part of a more generic theory of minimal Taylor clones recently developed over any finite domain (see \cite{barto24}).


\section*{Acknowledgements}
This research is supported by the Swiss National Science Foundation project 200020-169022 ``Lift and Project Methods for Machine Scheduling Through Theory and Experiments'' and 200021-207429 ``Ideal Membership Problems and the Bit Complexity of Sum of Squares Proofs''.

\bibliography{references}

\appendix

\section{Boolean Minority}




\subsection{An example}\label{sec:example}
We provide a simple example in \cref{tab:example} where we convert the reduced \GB basis in \lex order of a combinatorial ideal to one in \grlex order. Consider the problem formulated by the following (mod 2) equations: $x_1 \oplus x_3 \oplus x_4 = 0$ and $x_2 \oplus x_3 \oplus x_5 \oplus 1 = 0$. 

Note that $f_1 = x_3 \oplus x_4$, $f_2 = x_3 \oplus x_5 \oplus 1$, $f_3 = x_3$, $f_4 = x_4$ and $f_5 = x_5$. The reduced \GB basis in the \lex order is $G=G_1=\{x_1-M(f_1), x_2 - M(f_2)$, $x_3^2-x_3$, $x_4^2-x_4$, $x_5^2-x_5\}$. We start with $G_2=\LM(G_2)=\emptyset$, $B(G_2)=A=\{1\}$ (so $b_1=a_1=1$) and $q=x_5$. For the problem of $d=2$, we have

\begin{equation*}
    Q = \{x_5,x_4,x_3,x_2,x_1,x_5^2, x_5x_4, x_4^2, x_5x_3, x_4x_3, x_3^2, x_5x_2, \ldots, x_3x_1, x_2x_1, x_1^2\}.
\end{equation*}

We start with $q=x_5$ and since ${q}|_{G_1} = x_5$ and $x_5$ does not appear as the longest Boolean term of any element of $A$, $a_2=x_5$ is added to $A$ and $b_2=x_5$ is added to $B(G_2)$.
Then, $x_5$ is deleted from $Q$ and $x_4$ is the new $q$.
The iterations are similar for $q=x_4$ and $q=x_3$, so we have $b_3=a_3=x_4$ and $b_4=a_4=x_3$ and $x_4,x_3$ are deleted from $Q$. When $q=x_2$, we have ${q}|_{G_1} = f_2 = (x_3\oplus x_5 \oplus 1)$, and since the Boolean term does not appear in any $a\in A$, we add $a_5=(x_3\oplus x_5 \oplus 1)$ to $A$  and $b_5=x_2$ to $B(G_2)$.
For similar reasons, when $q=x_1$, we add $a_6= (x_3\oplus x_4) $ to $A$ and $b_6=x_1$ to $B(G_2)$.
\begin{table}[ht]
    \centering
    \begin{tabular}{c|c|c|c|c}
    \#  &  $q$ & $B(G_2)$ & $A$ & $G_2$\\
    \hline
    0 &  -  & 1 & $1$ & $\emptyset$ \\
    1 & $x_5$ & $x_5$ & $x_5$ & - \\
    2 & $x_4$ & $x_4$ & $x_4$ & -\\
    3\label{seehere} & $x_3$ & $x_3$ & $x_3$ & -\\
    4 & $x_2$ & $x_2$ & $x_3 \oplus x_5 \oplus 1$ & -\\
    5 & $x_1$ & $x_1$ & $x_3 \oplus x_4 $ & -\\
    6 & $x_5^2$ & - & - & $x_5^2-x_5$\\
    7 & $x_4x_5$ & $x_4x_5$ & $\frac{1}{2}[\reduce{x_4}{G_1}+\reduce{x_5}{G_1}$ & -\\
     &  &  & $-(x_4 \oplus x_5)]$ & \\
    8 & $x_4^2$ & - & - & $x_4^2-x_4$\\
    9 & $x_3x_5$ & - & - & $x_3x_5-\frac{1}{2}[x_2+x_3+x_5-1]$\\
    10 & $x_3x_4$ & - & - & $x_3x_4 - \frac{1}{2}[-x_1+x_3+x_4]$\\
    11 & $x_3^2$ & - & - & $x_3^2-x_3$\\
    12 & $x_2x_5$ & -  & - & $x_2x_5 - \frac{1}{2}[x_2 + x_3 + x_5 - 1]$\\
    13 & $x_2x_4$ & $x_2x_4$ & $\frac{1}{2}[{x_2}|_{G_1} + {x_4}|_{G_1}$ & -\\
     &  &  & $ - (x_3\oplus x_4 \oplus x_5 \oplus 1) ]$ & \\
    14 & $x_2x_3$ & - & - & $x_2x_3-\frac{1}{2}[x_2+x_3+x_5-1]$\\
    15 & $x_2^2$ & - & - & $x_2^2-x_2$\\
    16 & $x_1x_5$ & - & - & $x_1x_5+x_2x_4-\frac{1}{2}[x_1+x_2+x_4+x_5-1]$\\
    17 & $x_1x_4$ & - & - & $x_1x_4-\frac{1}{2}[x_1-x_3+x_4]$\\
    18 & $x_1x_3$ & - & - & $x_1x_3-\frac{1}{2}[x_1+x_3-x_4]$\\
    19 & $x_1x_2$ & - & - & $x_1x_2+x_4x_5-\frac{1}{2}[x_1+x_2+x_4+x_5-1]$\\
    20 & $x_1^2$ & - & - & $x_1^2-x_1$\\
\end{tabular}
    \caption{Example}
    \label{tab:example}
\end{table}
After the 5-th iteration (see \cref{tab:example}) is complete, we only have degree-two monomials in $Q$. When $q=x_5^2$ and ${q}|_{G_1} = \frac{1}{2}(x_5+x_5-0)=x_5$. Since $a_2=x_5$, $x_5$ appears as a Boolean term in $a_2$. Since the longest Boolean term appears already in $A$, ${q}|_{G_1}$ must be a linear combination of existing $\reduce{b_i}{G_1}$'s. That is to say,

\begin{equation*}
    {x_5^2}|_{G_1}=c_2 = \reduce{b_2}{G_1}={x_5}|_{G_1} \implies \reduce{x_5^2}{G_1}=\reduce{x_5}{G_1},
\end{equation*}

so the polynomial $x_5^2-x_5$ is added to $G_2$. Hence $x_5^2$ is added to $\LM(G_2)$ and $x_5^2$ is deleted from $Q$. When $q=x_5x_4$, 

\begin{equation*}
    {x_5x_4}|_{G_1}= f_4\cdot f_5 =\frac{1}{2}[x_4+x_5-(x_4 \oplus x_5)]=\frac{1}{2}[\reduce{x_4}{G_1}+\reduce{x_5}{G_1}-(x_4 \oplus x_5)].
\end{equation*}

The longest Boolean term of ${q}|_{G_1}$ is $(x_4\oplus x_5)$ which does not appear in any $a\in A$, so
$a_7 = 1/2[\reduce{x_4}{G_1}+\reduce{x_5}{G_1}-(x_4 \oplus x_5)]$ is added to $A$ and $b_7 = x_5x_4$ is added to $B(G_2)$.
When $q=x_4^2$ (this is similar to the case when $q=x_5^2$), we see that $x_4^2-x_4$ is added to $G_2$ and $x_4^2$ to $\LM(G_2)$. When  $q=x_3x_5$,

\begin{equation*}
    {x_3x_5}|_{G_1}= f_3\cdot f_5=\frac{1}{2}[x_3+x_5-(x_3 \oplus x_5)].
\end{equation*}

Note that $(x_3\oplus x_5 \oplus 1)$ appears in $a_5\in C$. We use the fact that $(f\oplus 1) = 1 - f$ (see \cref{alg:conversion}), and we have

\begin{align*}
    {x_3x_5}|_{G_1}&= \frac{1}{2}[x_3+x_5-(x_3 \oplus x_5)] =\frac{1}{2}[\reduce{x_3}{G_1}+\reduce{x_5}{G_1}-(1 - (x_3 \oplus x_5\oplus 1))]\\
    &= \frac{1}{2}[\reduce{x_2}{G_1}+\reduce{x_3}{G_1}+\reduce{x_5}{G_1}-\reduce{1}{G_1}]
\end{align*}

and thus $x_3x_5-\frac{1}{2}[x_2+x_3+x_5-1]$ is added to $G_2$ and $x_3x_5$ to $\LM(G_2)$. The rest of the polynomials in $B(G_2),G_2,A$ are as shown in \cref{tab:example}. After the 20-th iteration, $Q$ becomes empty and algorithm halts. This gives the 2-truncated reduced \GB basis $G_2$ of the combinatorial ideal. Note that this is in fact the reduced \GB basis in its entirety for this particular example.

\end{document}

%% file: comandi.tex
\usepackage{amsthm}
\usepackage{amsmath}
\usepackage{mathrsfs}
\usepackage{amssymb}
\usepackage{verbatim}

\usepackage{algorithmic}
\usepackage[linesnumbered,ruled,vlined]{algorithm2e}
\usepackage{color}
\usepackage{mdframed}
\usepackage{marginnote}

\usepackage{enumerate}
\usepackage{array}
\usepackage[a4paper, total={6in, 8in}]{geometry}
\usepackage{mdframed}
\usepackage{marginnote}
\usepackage{amsfonts}
\usepackage[mathlines]{lineno}
\usepackage{subfig}
\usepackage{graphicx}
\usepackage[shortlabels]{enumitem}
\usepackage{url}
\usepackage[pdftex, colorlinks,
  linkcolor=blue,
  citecolor=blue,
  filecolor=blue,urlcolor=blue]%
{hyperref}
  
\usepackage[capitalise]{cleveref}

\newcommand{\Mod}[1]{\ (\mathrm{mod}\ #1)}

\newcommand{\sos}{\text{\sc{SoS}}}

\newcommand{\Cc}{\mathcal{C}}

\newcommand{\Ideal}[1]{{\textbf{I}}\left( #1 \right)}
\newcommand{\GIdeal}[1]{\left\langle #1 \right\rangle}

\newcommand{\CSP}{\textsc{CSP}}
\newcommand{\IMP}{\textsc{IMP}}

\newcommand{\CPC}{\textsc{CPC}}
\newcommand{\CF}{\textsc{CF}}

\newcommand{\lex}{\textsf{lex }}

\newcommand{\grlex}{\textsf{grlex }}

\newcommand{\Variety}[1]{{\textbf{V}}\left( #1 \right)}
\newcommand{\spn}[1]{\left\langle #1 \right\rangle}

\newcommand{\I}{\emph{\texttt{I}}}
\newcommand{\Zz}{\mathbb{Z}}

\newcommand{\multideg}{\textnormal{multideg}}
\newcommand{\LM}{\textnormal{LM}}
\newcommand{\LT}{\textnormal{LT}}
\newcommand{\LC}{\textnormal{LC}}
\newcommand{\LCM}{\textnormal{lcm}}
\newcommand{\GB}{\text{Gr\"{o}bner} }
\newcommand{\reduce}[2]{{#1}|_{#2}}

\newcommand{\Field}{\mathbb{F}}
\newcommand{\Real}{\mathbb{R}}

\newcommand{\ca}[1]{\mathcal{#1}}

\newtheorem*{theorem*}{Theorem}

\providecommand{\keywords}[1]
{
  \small	
  \textbf{\textit{Keywords---}} #1
}
\providecommand{\subjectclass}[1]
{
  \small	
  \textbf{2010 MSC}\quad #1
}

%% file: main.bbl
\begin{thebibliography}{10}

\bibitem{barto24}
Libor Barto, Zarathustra Brady, Andrei Bulatov, Marcin Kozik, and Dmitriy Zhuk.
\newblock Unifying the {T}hree {A}lgebraic {A}pproaches to the {CSP} via
  {M}inimal {T}aylor {A}lgebras.
\newblock {\em TheoretiCS}, 3, 2024.

\bibitem{barto_et_al:DFU:2017:6959}
Libor Barto, Andrei Krokhin, and Ross Willard.
\newblock {Polymorphisms, and How to Use Them}.
\newblock In Andrei Krokhin and Stanislav Zivny, editors, {\em The Constraint
  Satisfaction Problem: Complexity and Approximability}, volume~7 of {\em
  Dagstuhl Follow-Ups}, pages 1--44. Schloss Dagstuhl--Leibniz-Zentrum fuer
  Informatik, Dagstuhl, Germany, 2017.
\newblock URL: \url{http://drops.dagstuhl.de/opus/volltexte/2017/6959}, \href
  {https://doi.org/10.4230/DFU.Vol7.15301.1}
  {\path{doi:10.4230/DFU.Vol7.15301.1}}.

\bibitem{BeameIKPP94}
Paul Beame, Russell Impagliazzo, Jan Kraj{\'{\i}}cek, Toniann Pitassi, and
  Pavel Pudl{\'{a}}k.
\newblock {Lower Bound on {H}ilbert's {N}ullstellensatz and propositional
  proofs}.
\newblock In {\em 35th Annual Symposium on Foundations of Computer Science,
  Santa Fe, New Mexico, USA, 20-22 November 1994}, pages 794--806, 1994.

\bibitem{bharathi_et_al:MFCS2020}
Arpitha~P. Bharathi and Monaldo Mastrolilli.
\newblock {Ideal Membership Problem and a Majority Polymorphism over the
  Ternary Domain}.
\newblock In Javier Esparza and Daniel Kr{\'a}ľ, editors, {\em 45th
  International Symposium on Mathematical Foundations of Computer Science (MFCS
  2020)}, volume 170 of {\em Leibniz International Proceedings in Informatics
  (LIPIcs)}, pages 13:1--13:13, Dagstuhl, Germany, 2020. Schloss
  Dagstuhl--Leibniz-Zentrum f{\"u}r Informatik.
\newblock URL: \url{https://drops.dagstuhl.de/opus/volltexte/2020/12682}, \href
  {https://doi.org/10.4230/LIPIcs.MFCS.2020.13}
  {\path{doi:10.4230/LIPIcs.MFCS.2020.13}}.

\bibitem{bharathi2020ideal}
Arpitha~P. Bharathi and Monaldo Mastrolilli.
\newblock {Ideal Membership Problem for Boolean Minority}, 2020.
\newblock \href {https://arxiv.org/abs/2006.16422} {\path{arXiv:2006.16422}}.

\bibitem{bharathi_et_al:MFCS2021}
Arpitha~P. Bharathi and Monaldo Mastrolilli.
\newblock {Ideal Membership Problem for Boolean Minority and Dual
  Discriminator}.
\newblock In Filippo Bonchi and Simon~J. Puglisi, editors, {\em 46th
  International Symposium on Mathematical Foundations of Computer Science (MFCS
  2021)}, volume 202 of {\em Leibniz International Proceedings in Informatics
  (LIPIcs)}, pages 16:1--16:20, Dagstuhl, Germany, 2021. Schloss Dagstuhl --
  Leibniz-Zentrum f{\"u}r Informatik.
\newblock URL: \url{https://drops.dagstuhl.de/opus/volltexte/2021/14456}, \href
  {https://doi.org/10.4230/LIPIcs.MFCS.2021.16}
  {\path{doi:10.4230/LIPIcs.MFCS.2021.16}}.

\bibitem{Bharathi_etal_SIAM2022}
Arpitha~P. Bharathi and Monaldo Mastrolilli.
\newblock Ideal membership problem over 3-element {CSP}s with dual
  discriminator polymorphism.
\newblock {\em SIAM Journal on Discrete Mathematics}, 36(3):1800--1822, 2022.
\newblock \href {https://arxiv.org/abs/https://doi.org/10.1137/21M1397131}
  {\path{arXiv:https://doi.org/10.1137/21M1397131}}, \href
  {https://doi.org/10.1137/21M1397131} {\path{doi:10.1137/21M1397131}}.

\bibitem{brady2022notes}
Zarathustra Brady.
\newblock Notes on {CSP}s and polymorphisms, 2022.
\newblock \href {https://arxiv.org/abs/2210.07383} {\path{arXiv:2210.07383}}.

\bibitem{BUCHBERGER2006475}
Bruno Buchberger.
\newblock {Bruno Buchberger’s PhD thesis 1965: An algorithm for finding the
  basis elements of the residue class ring of a zero dimensional polynomial
  ideal}.
\newblock {\em Journal of Symbolic Computation}, 41(3):475 -- 511, 2006.
\newblock Logic, Mathematics and Computer Science: Interactions in honor of
  Bruno Buchberger (60th birthday).
\newblock URL:
  \url{http://www.sciencedirect.com/science/article/pii/S0747717105001483},
  \href {https://doi.org/10.1016/j.jsc.2005.09.007}
  {\path{doi:10.1016/j.jsc.2005.09.007}}.

\bibitem{Bulatov17}
Andrei~A. Bulatov.
\newblock A dichotomy theorem for nonuniform {CSP}s (best paper award).
\newblock In {\em 58th {IEEE} Annual Symposium on Foundations of Computer
  Science, {FOCS} 2017, Berkeley, CA, USA, October 15-17, 2017}, pages
  319--330, 2017.

\bibitem{2017dfu7}
Andrei~A. Bulatov.
\newblock Constraint satisfaction problems: Complexity and algorithms.
\newblock {\em ACM SIGLOG News}, 5(4):4--24, November 2018.
\newblock URL: \url{http://doi.acm.org/10.1145/3292048.3292050}, \href
  {https://doi.org/10.1145/3292048.3292050}
  {\path{doi:10.1145/3292048.3292050}}.

\bibitem{bulatov2020complexity}
Andrei~A. Bulatov and Akbar Rafiey.
\newblock {On the Complexity of {CSP}-based Ideal Membership Problems}, 2020.
\newblock URL: \url{https://arxiv.org/abs/2011.03700}, \href
  {https://doi.org/10.48550/ARXIV.2011.03700}
  {\path{doi:10.48550/ARXIV.2011.03700}}.

\bibitem{bulatov2022STOCupdated}
Andrei~A. Bulatov and Akbar Rafiey.
\newblock {On the Complexity of {CSP}-Based Ideal Membership Problems}.
\newblock In {\em Proceedings of the 54th Annual ACM SIGACT Symposium on Theory
  of Computing}, STOC 2022, page 436–449, New York, NY, USA, 2022.
  Association for Computing Machinery.
\newblock \href {https://doi.org/10.1145/3519935.3520063}
  {\path{doi:10.1145/3519935.3520063}}.

\bibitem{bulatov_abeliangroups}
Andrei~A. Bulatov and Akbar Rafiey.
\newblock {The Ideal Membership Problem and Abelian Groups}, 2022.
\newblock URL: \url{https://arxiv.org/abs/2201.05218}, \href
  {https://doi.org/10.48550/ARXIV.2201.05218}
  {\path{doi:10.48550/ARXIV.2201.05218}}.

\bibitem{BussP98}
Samuel~R. Buss and Toniann Pitassi.
\newblock {Good Degree Bounds on {N}ullstellensatz Refutations of the Induction
  Principle}.
\newblock {\em J. Comput. Syst. Sci.}, 57(2):162--171, 1998.

\bibitem{Chen09}
Hubie Chen.
\newblock A rendezvous of logic, complexity, and algebra.
\newblock {\em ACM Comput. Surv.}, 42(1):2:1--2:32, December 2009.
\newblock URL: \url{http://doi.acm.org/10.1145/1592451.1592453}, \href
  {https://doi.org/10.1145/1592451.1592453}
  {\path{doi:10.1145/1592451.1592453}}.

\bibitem{Jeavons1994_Characterising_tract_constraints}
Martin~C. Cooper, David~A. Cohen, and Peter~G. Jeavons.
\newblock Characterising tractable constraints.
\newblock {\em Artificial Intelligence}, 65(2):347--361, 1994.
\newblock URL:
  \url{https://www.sciencedirect.com/science/article/pii/0004370294900213},
  \href {https://doi.org/10.1016/0004-3702(94)90021-3}
  {\path{doi:10.1016/0004-3702(94)90021-3}}.

\bibitem{Cox}
David~A. Cox, John Little, and Donal O'Shea.
\newblock {\em Ideals, Varieties, and Algorithms: An Introduction to
  Computational Algebraic Geometry and Commutative Algebra}.
\newblock Springer Publishing Company, Incorporated, 4th edition, 2015.

\bibitem{FAUGERE1993329}
Jean-Charles Faug{\`e}re, Patrizia~M. Gianni, Daniel Lazard, and Teo Mora.
\newblock {Efficient Computation of Zero-dimensional {G}r{\"o}bner Bases by
  Change of Ordering}.
\newblock {\em Journal of Symbolic Computation}, 16(4):329 -- 344, 1993.
\newblock URL:
  \url{http://www.sciencedirect.com/science/article/pii/S0747717183710515},
  \href {https://doi.org/10.1006/jsco.1993.1051}
  {\path{doi:10.1006/jsco.1993.1051}}.

\bibitem{goemans1995improved}
Michel~X Goemans and David~P Williamson.
\newblock Improved approximation algorithms for maximum cut and satisfiability
  problems using semidefinite programming.
\newblock {\em Journal of the ACM (JACM)}, 42(6):1115--1145, 1995.

\bibitem{Grigoriev98}
Dima Grigoriev.
\newblock {Tseitin's Tautologies and Lower Bounds for {N}ullstellensatz
  Proofs}.
\newblock In {\em 39th Annual Symposium on Foundations of Computer Science,
  {FOCS} '98, November 8-11, 1998, Palo Alto, California, {USA}}, pages
  648--652, 1998.

\bibitem{HilbertBasisTheorem}
David Hilbert.
\newblock Ueber die theorie der algebraischen formen.
\newblock {\em Mathematische Annalen}, 36:473 -- 534, 1890.
\newblock URL: \url{http://eudml.org/doc/157652}, \href
  {https://doi.org/10.1007/BF01208503} {\path{doi:10.1007/BF01208503}}.

\bibitem{Hilbert1893}
David Hilbert.
\newblock Ueber die vollen invariantensysteme.
\newblock {\em Mathematische Annalen}, 42:313--373, 1893.
\newblock URL: \url{http://eudml.org/doc/157652}.

\bibitem{JEAVONS1998185}
Peter Jeavons.
\newblock On the algebraic structure of combinatorial problems.
\newblock {\em Theoretical Computer Science}, 200(1):185 -- 204, 1998.
\newblock URL:
  \url{http://www.sciencedirect.com/science/article/pii/S0304397597002302},
  \href {https://doi.org/10.1016/S0304-3975(97)00230-2}
  {\path{doi:10.1016/S0304-3975(97)00230-2}}.

\bibitem{Jeavons:1997:CPC}
Peter Jeavons, David Cohen, and Marc Gyssens.
\newblock {Closure Properties of Constraints}.
\newblock {\em J. ACM}, 44(4):527--548, July 1997.
\newblock URL: \url{http://doi.acm.org/10.1145/263867.263489}, \href
  {https://doi.org/10.1145/263867.263489} {\path{doi:10.1145/263867.263489}}.

\bibitem{Laurent2009}
Monique Laurent.
\newblock {\em Sums of Squares, Moment Matrices and Optimization Over
  Polynomials}, pages 157--270.
\newblock Springer New York, New York, NY, 2009.
\newblock \href {https://doi.org/10.1007/978-0-387-09686-5_7}
  {\path{doi:10.1007/978-0-387-09686-5_7}}.

\bibitem{MACKWORTH1977}
Alan~K. Mackworth.
\newblock Consistency in networks of relations.
\newblock {\em Artificial Intelligence}, 8(1):99--118, 1977.
\newblock URL:
  \url{https://www.sciencedirect.com/science/article/pii/0004370277900078},
  \href {https://doi.org/10.1016/0004-3702(77)90007-8}
  {\path{doi:10.1016/0004-3702(77)90007-8}}.

\bibitem{mastrolilli_talg21}
Monaldo Mastrolilli.
\newblock The complexity of the ideal membership problem and theta bodies for
  constrained problems over the boolean domain.
\newblock {\em CoRR, to appear in ACM Transactions on Algorithms},
  abs/1904.04072, 2019.
\newblock URL: \url{http://arxiv.org/abs/1904.04072}, \href
  {https://arxiv.org/abs/1904.04072} {\path{arXiv:1904.04072}}.

\bibitem{MonaldoMastrolilli2019}
Monaldo Mastrolilli.
\newblock {The Complexity of the Ideal Membership Problem for Constrained
  Problems over the Boolean Domain}.
\newblock In {\em Proceedings of the Thirtieth Annual ACM-SIAM Symposium on
  Discrete Algorithms}, SODA '19, pages 456--475, Philadelphia, PA, USA, 2019.
  Society for Industrial and Applied Mathematics.
\newblock URL: \url{http://dl.acm.org/citation.cfm?id=3310435.3310464}.

\bibitem{Mayr1989}
Ernst~W. Mayr.
\newblock {Membership in polynomial ideals over Q is exponential space
  complete}.
\newblock In B.~Monien and R.~Cori, editors, {\em STACS 89}, pages 400--406,
  Berlin, Heidelberg, 1989. Springer Berlin Heidelberg.

\bibitem{MAYR1982305}
Ernst~W. Mayr and Albert~R. Meyer.
\newblock The complexity of the word problems for commutative semigroups and
  polynomial ideals.
\newblock {\em Advances in Mathematics}, 46(3):305--329, 1982.
\newblock URL:
  \url{http://www.sciencedirect.com/science/article/pii/0001870882900482},
  \href {https://doi.org/10.1016/0001-8708(82)90048-2}
  {\path{doi:10.1016/0001-8708(82)90048-2}}.

\bibitem{odonnell2017}
Ryan O'Donnell.
\newblock {SOS Is Not Obviously Automatizable, Even Approximately}.
\newblock In Christos~H. Papadimitriou, editor, {\em 8th Innovations in
  Theoretical Computer Science Conference (ITCS 2017)}, volume~67 of {\em
  Leibniz International Proceedings in Informatics (LIPIcs)}, pages
  59:1--59:10, Dagstuhl, Germany, 2017. Schloss Dagstuhl--Leibniz-Zentrum fuer
  Informatik.
\newblock URL: \url{http://drops.dagstuhl.de/opus/volltexte/2017/8198}, \href
  {https://doi.org/10.4230/LIPIcs.ITCS.2017.59}
  {\path{doi:10.4230/LIPIcs.ITCS.2017.59}}.

\bibitem{Pixley1975}
Alden~F. Pixley and Kirby~A. Baker.
\newblock Polynomial interpolation and the chinese remainder theorem for
  algebraic systems.
\newblock {\em Mathematische Zeitschrift}, 143:165--174, 1975.
\newblock URL: \url{http://eudml.org/doc/172222}.

\bibitem{raghavendra_weitz2017}
Prasad Raghavendra and Benjamin Weitz.
\newblock {On the Bit Complexity of Sum-of-Squares Proofs}.
\newblock In Ioannis Chatzigiannakis, Piotr Indyk, Fabian Kuhn, and Anca
  Muscholl, editors, {\em 44th International Colloquium on Automata, Languages,
  and Programming (ICALP 2017)}, volume~80 of {\em Leibniz International
  Proceedings in Informatics (LIPIcs)}, pages 80:1--80:13, Dagstuhl, Germany,
  2017. Schloss Dagstuhl--Leibniz-Zentrum fuer Informatik.
\newblock URL: \url{http://drops.dagstuhl.de/opus/volltexte/2017/7380}, \href
  {https://doi.org/10.4230/LIPIcs.ICALP.2017.80}
  {\path{doi:10.4230/LIPIcs.ICALP.2017.80}}.

\bibitem{RaghavendraW17}
Prasad Raghavendra and Benjamin Weitz.
\newblock On the bit complexity of sum-of-squares proofs.
\newblock In {\em 44th International Colloquium on Automata, Languages, and
  Programming, {ICALP}, Poland}, pages 80:1--80:13, 2017.

\bibitem{Schaefer78}
Thomas~J. Schaefer.
\newblock {The Complexity of Satisfiability Problems}.
\newblock In {\em Proceedings of the Tenth Annual ACM Symposium on Theory of
  Computing}, STOC '78, pages 216--226, New York, NY, USA, 1978. ACM.
\newblock URL: \url{http://doi.acm.org/10.1145/800133.804350}, \href
  {https://doi.org/10.1145/800133.804350} {\path{doi:10.1145/800133.804350}}.

\bibitem{PIT}
Amir Shpilka.
\newblock Recent results on polynomial identity testing.
\newblock In Alexander Kulikov and Nikolay Vereshchagin, editors, {\em Computer
  Science -- Theory and Applications}, pages 397--400, Berlin, Heidelberg,
  2011. Springer Berlin Heidelberg.

\bibitem{szendrei}
\'{A}gnes Szendrei.
\newblock {\em Clones in universal algebra}.
\newblock Les Presses de l'Universit\'{e} de Montr\'{e}al, 1986.

\bibitem{vandongenPhd}
Marc~R.C. van Dongen.
\newblock {\em Constraints, Varieties, and Algorithms}.
\newblock PhD thesis, Department of Computer Science, University College, Cork,
  Ireland, 2002.
\newblock URL: \url{http://csweb.ucc.ie/~dongen/papers/UCC/02/thesis.pdf}.

\bibitem{Weitz:Phd}
Benjamin Weitz.
\newblock {\em {Polynomial Proof Systems, Effective Derivations, and their
  Applications in the Sum-of-Squares Hierarchy}}.
\newblock PhD thesis, EECS Department, University of California, Berkeley, May
  2017.
\newblock URL:
  \url{http://www2.eecs.berkeley.edu/Pubs/TechRpts/2017/EECS-2017-38.html}.

\bibitem{Zhuk17}
Dmitriy Zhuk.
\newblock A proof of the {CSP} dichotomy conjecture.
\newblock {\em J. {ACM}}, 67(5):30:1--30:78, 2020.
\newblock \href {https://doi.org/10.1145/3402029} {\path{doi:10.1145/3402029}}.

\end{thebibliography}
